\documentclass[12pt]{article}

\usepackage{amsmath,amsthm,amssymb,amscd}

\usepackage{graphicx}
\usepackage{epsfig}
\usepackage[hidelinks]{hyperref}
\usepackage{harvard}
\usepackage{rotating}
\usepackage{pww}

\usepackage{enumitem} 
\usepackage{caption, subcaption} 
\usepackage{capt-of} 
\usepackage[flushleft]{threeparttable} 
\usepackage{multirow}
\usepackage{booktabs}

\usepackage{xcolor,soul} 
\usepackage[linesnumbered,ruled,vlined]{algorithm2e} 

\newtheorem{Theorem}{Theorem} 

\author{Sheng Dai
\footnote
{
Department of Information and Service Management, Aalto University School of Business, 02150 Espoo, Finland. Email: \texttt{sheng.dai@aalto.fi}.
}
}

\title{\bf Variable selection in convex quantile regression: $\mathcal{L}_1$-norm or $\mathcal{L}_0$-norm regularization?}
\date{July 2021}

\begin{document}
\captionsetup[figure]{labelfont={bf},labelformat={default},labelsep=period,name={Fig.}}
\captionsetup[table]{labelfont={bf},labelformat={default},labelsep=period,name={Table}}

\citationmode{abbr}
\bibliographystyle{jbes}

\maketitle
 
\vfill
\vfill

\begin{abstract}
The curse of dimensionality is a recognized challenge in nonparametric estimation. This paper develops a new $\mathcal{L}_0$-norm regularization approach to the convex quantile and expectile regressions for subset variable selection. We show how to use mixed integer programming to solve the proposed $\mathcal{L}_0$-norm regularization approach in practice and build a link to the commonly used $\mathcal{L}_1$-norm regularization approach. A Monte Carlo study is performed to compare the finite sample performances of the proposed $\mathcal{L}_0$-penalized convex quantile and expectile regression approaches with the $\mathcal{L}_1$-norm regularization approaches. The proposed approach is further applied to benchmark the sustainable development performance of the OECD countries and empirically analyze the accuracy in the dimensionality reduction of variables. The results from the simulation and application illustrate that the proposed $\mathcal{L}_0$-norm regularization approach can more effectively address the curse of dimensionality than the $\mathcal{L}_1$-norm regularization approach in multidimensional spaces.
\\[5mm]
\end{abstract}
\vfill

\thispagestyle{empty}
\newpage
\setcounter{page}{1}
\setcounter{footnote}{0}
\pagenumbering{arabic}
\baselineskip 20pt


\section{Introduction}\label{sec:intro}

The curse of dimensionality, a well-known problem in statistics and econometrics, refers to the poor performance of nonparametric methods in a high-dimensional data space (\citename{Lavergne2008}, \citeyear*{Lavergne2008}). The nonparametric estimators are weakened in prediction accuracy and exploratory power as the dimension increases due to the sparsity of data (\citename{Stone2007}, \citeyear*{Stone2007}). This problem has motivated a wide range of literature over the past decades addressing different dimensionality-reduction methods, which can be utilized to reduce the effects of dimensionality  (e.g., \citename{Sinha2015}, \citeyear*{Sinha2015}; \citename{Wilson2018}, \citeyear*{Wilson2018}; \citename{Lee2020}, \citeyear*{Lee2020}; and references therein). However, in the production efficiency analysis area, the curse of dimensionality remains an unresolved issue for nonparametric estimators such as data envelopment analysis (DEA), free disposal hull (FDH), and sign-constrained convex nonparametric least squares (SCNLS) (\citename{Lee2020}, \citeyear*{Lee2020}).

These nonparametric estimators would lead to an overfitting problem when faced with a limited data sample size in multidimensional spaces, especially when the estimated production function approaches the boundary of the convex hull of the production set (\citename{Mazumder2019}, \citeyear*{Mazumder2019}). Accordingly, observations are likely to lie close to or on the production frontier and become more efficient, presumably jeopardizing the accuracy of production function estimation. To date, no prominent method has been made available to directly circumvent these problems caused by the curse of dimensionality.

To address the curse of dimensionality in production models, we could consider either increasing sample sizes or reducing dimensions. Increasing sample sizes aims to resatisfy a ``rules of thumb” relation among the number of observations ($n$) and the number of inputs ($d$) and outputs ($q$), e.g., $n \ge 2(d+q)$ proposed in \citeasnoun{Homburg2001} and $n \ge 2d*q$ proposed in \citeasnoun{Dyson2001}. However, we cannot always enlarge sample sizes in actual applications. Moreover, increasing sample sizes could cause a large computational burden (see, e.g., \citename{Dula2011}, \citeyear*{Dula2011}; \citename{Lee2013}, \citeyear*{Lee2013}). For instance, for the evaluation of the Sustainable Development Goals (SDGs), even when using classical production models, the input-output variables have substantial multicollinearity (e.g., \citename{Lamichhane2020}, \citeyear*{Lamichhane2020}; \citename{Dai2020}, \citeyear*{Dai2020}). It is not applicable to reduce multicollinearity by increasing the number of observations. Thus, increasing sample sizes is not an optimal choice. 

By contrast, dimensionality reduction, which includes feature extraction (i.e., variable extraction) and feature selection (i.e., variable selection) techniques, is a more appealing approach. A vast body of literature on feature extraction for DEA has emerged. One typical choice is to use principal component analysis (PCA) to reduce the dimensionality of DEA production models (e.g., \citename{Adler2010a}, \citeyear*{Adler2010a}; \citename{Nataraja2011}, \citeyear*{Nataraja2011}; \citename{Benitez-Pena2019}, \citeyear*{Benitez-Pena2019}). Another commonly used method is to utilize the multicollinearity among inputs or outputs (\citename{Wilson2018}, \citeyear*{Wilson2018}). However, using the new artificial variables created by the variable extraction technique instead of the original variables cannot lead to the formation of a meaningful DEA-estimated frontier (\citename{Lee2020}, \citeyear*{Lee2020}). Furthermore, it is difficult to identify the production transformation process from the new artificial inputs to the output. To better understand the production transformation process, variable extraction is obviously not a main technique to reduce the dimensionality in production models.

In addition to variable extraction, variable selection is commonly seen in the context of regression. According to a review on variable selection methods by \citeasnoun{Sinha2015}, the mainstream approaches in the statistical modeling and machine learning fields include regularization, stepwise selection, best subset, heuristic algorithms, and Bayesian model averaging. The pros and cons of each approach are well addressed in \citeasnoun{Sinha2015} in detail. More importantly, most of these approaches have been widely applied in production efficiency analysis to handle dimensionality reduction and variable selection.

Stepwise selection (i.e., forward selection and backward elimination) is most frequently used as a variable selection technique in production models, implementing an ex post analysis of the sensitivity of the production frontier to additional variables and detecting whether the relevant variables should be included or removed (e.g., \citename{Pastor2002}, \citeyear*{Pastor2002}; \citename{Wagner2007}, \citeyear*{Wagner2007}). However, the relevance among variables in stepwise selection will affect the final efficiency of the frontier estimation (\citename{Li2017}, \citeyear*{Li2017}). The least absolute shrinkage and selection operator (Lasso)-type regularization (i.e., $\mathcal{L}_1$-norm regularization) is another common technique for choosing a subset of variables (e.g., \citename{Qin2014}, \citeyear*{Qin2014}; \citename{Lee2020}, \citeyear*{Lee2020}; \citename{Chen2021}, \citeyear*{Chen2021}). A few recent attempts have been made to apply integer programming (i.e., a variant of $\mathcal{L}_0$-norm regularization) to achieve direct control over the number of subset variables instead of seeking sparsity in $\mathcal{L}_1$-norm regularization (\citename{Keshvari2018}, \citeyear*{Keshvari2018}; \citename{Benitez-Pena2019}, \citeyear*{Benitez-Pena2019}). However, the existing literature on production models does not address the question of which regularization would be a better subset selection approach via, e.g., an effectiveness comparison.
 Therefore, the effectiveness comparison between the $\mathcal{L}_1$- and $\mathcal{L}_0$-norm regularization approaches in the dimensionality reduction of variables would warrant further research.

This study is also motivated by the application to SDG evaluation. The 2030 Agenda for Sustainable Development, adopted by all 193 United Nations member countries in 2015, 
is a widely acknowledged urgent call for actions with the sustainable aims to end poverty, reduce inequality, boost economic growth and tackle climate change.\footnote{United Nations, The 17 Goals. See more details regarding the SDGs at https://sdgs.un.org/goals.
}
The core of Agenda 2030 comprises the 17 SDGs and 169 targets, representing a global consensus. In addition, under the global indicator framework,\footnote{
United Nations, Global indicator framework for the Sustainable Development Goals and targets of the 2030 Agenda for Sustainable Development. Available at https://unstats.un.org/sdgs/indicators/indicators-list/.
}
231 unique indicators have been designed to measure these targets and monitor progress toward the achievement of the SDGs. Consequently, the integrated assessment approaches based on this high number of indicators are immensely challenging to implement due to the correlation among the indicators and the unit of measurement differences (\citename{Lamichhane2020}, \citeyear*{Lamichhane2020}). The need and challenge to better benchmark the degree of sustainable development of countries triggered the emerging integration of SDG assessment and production economics (e.g., \citename{Singpai}, \citeyear*{Singpai}). 

Furthermore, the heterogeneity regarding the SDG assessment has become another pervasive concern throughout the production function estimation process. Very large differences can be noted in the level, the stage, and the goal of development worldwide. It is thus improper to evaluate the degree of sustainable development of countries following a uniform standard. However, convex quantile regression (CQR) (\citename{Wang2014c}, \citeyear*{Wang2014c}; \citename{Kuosmanen2015d}, \citeyear*{Kuosmanen2015d}) is an alternative that can be used to overcome heterogeneity and that is more robust to outliers, providing an overall picture of the conditional distributions at any given quantile. In contrast to the conventional full frontier estimation approach, CQR can evaluate the actual level of sustainable development locally in the interior of the production possibility set.

In this paper, we have two main purposes: 1) We aim to compare the $\mathcal{L}_1$-norm regularization with $\mathcal{L}_0$-norm regularization and determine which would outperform in subset selection in the context of the production model. In doing so, a Monte Carlo (MC) study and an empirical application are performed to compare the effectiveness of methods in reducing dimensionality. 2) We aim to evaluate the progress of SDGs for OECD countries and to determine what causes the inequality in sustainable development. Hence, extensions to convex quantile and expectile regression using the $\mathcal{L}_1$- and $\mathcal{L}_0$-norm regularization have been developed and applied. The main contributions are thus summarized below.
\begin{enumerate}[label= \arabic*), leftmargin=2\parindent]
    \setlength{\itemsep}{1pt}
    \setlength{\parskip}{0pt}
    \setlength{\parsep}{0pt}
    \item $\mathcal{L}_0$-penalized CQR. To reduce the dimensionality of the production model and increase the robustness to heterogeneity in SDG assessment, we develop a new $\mathcal{L}_0$-norm regularization approach to convex quantile and expectile regression (\citename{Wang2014c}, \citeyear*{Wang2014c}; \citename{Kuosmanen2015d}, \citeyear*{Kuosmanen2015d}), show how to solve it in practice using the mixed integer programming (MIP) method, and build a link with the commonly used $\mathcal{L}_1$-norm regularization approach. By taking advantage of $\mathcal{L}_0$-norm regularization and quantile estimation, we can make full use of each observation and fairly evaluate the degree of SDGs for OECD countries even in a small sample size.
    \item Testing the effectiveness of $\mathcal{L}_1$- and $\mathcal{L}_0$-norm regularization in dimensionality reduction within the context of the production model. $\mathcal{L}_0$-norm regularization directly controls the number of subset variables rather than the level of sparsity uniformly. Does this mean that $\mathcal{L}_0$-norm regularization is a better choice for subset selection? It is thus natural to compare the performance of the proposed $\mathcal{L}_0$-norm regularization with that of the competing $\mathcal{L}_1$-norm regularization. In the MC and application sections, we show that the proposed $\mathcal{L}_0$-norm regularization outperforms in terms of dimensionality reduction and that $\mathcal{L}_1$-norm regularization cannot implicitly eliminate the input variables in the regression-based production function estimation.
    \item Application to SDG assessment in OECD countries. There are few SDG evaluations in the context of production economics. We present an empirical assessment in Section \ref{sec:sdg} from the perspective of production function estimation using the OECD country data from the SDG indicators database.
\end{enumerate}

Furthermore, the penalized convex quantile and expectile regression approaches would be greatly subject to the computational burden from the \textit{O}($n^2$) linear constraints as the number of observations increases. To mitigate the computational burden, we improve upon the base of the cutting-plane algorithm (\citename{Bertsimas2020}, \citeyear*{Bertsimas2020}) to solve the $\mathcal{L}_1$- and $\mathcal{L}_0$-penalized convex quantile and expectile regression approaches. An additional experiment is performed to justify the superiority in performance of the new algorithm over the CNLS-G algorithm proposed by \citeasnoun{Lee2013}. 

The rest of the paper is organized as follows. Section~\ref{sec:meth} introduces the $\mathcal{L}_1$- and $\mathcal{L}_0$-penalized convex quantile and expectile regression approaches. Section~\ref{sec:mc} performs an MC study to compare the finite sample performances of the $\mathcal{L}_1$- and $\mathcal{L}_0$-norm regularization approaches from two metrics. An empirical application to SGD assessment is prepared in Section~\ref{sec:sdg}. Section~\ref{sec:conc} concludes this paper. An introduction to the CNLS-A algorithm and additional tables and figures are provided in the Appendix \ref{sec:anlsa}.


\section{Methodology}\label{sec:meth}
\setcounter{equation}{0}

\subsection{Penalized quantile production function}

Consider the nonparametric conditional quantile production functions (e.g., \citename{Dai2020}, \citeyear*{Dai2020}),
\begin{equation}
Q_{y_i}(\tau \, | \, \bx_i) = f_\tau(\bx_i) + F^{-1}_{\varepsilon_i}(\tau), \quad \text{for} \quad i = 1, \cdots, n
\label{eq:qpf}
\end{equation}
where $\tau \in (0, 1)$ denotes the quantile order, $y_i \in \real$ and $\bx_i \in \real^d$ are the output and input variables, respectively, $\varepsilon_i$ denotes the independent and identically distributed errors, and $F_{\varepsilon_i}$ represents the cumulative distribution function of the errors. 

The nonparametric quantile production model \eqref{eq:qpf} does not require the \textit{a priori} functional form for $f_\tau(\cdot):\real^d \xrightarrow{} \real$ but instead assumes that $f_\tau(\cdot)$ satisfies certain axiomatic properties (e.g., monotonicity, concavity). As such, one can readily resort to CQR (\citename{Wang2014c}, \citeyear*{Wang2014c}; \citename{Kuosmanen2015d}, \citeyear*{Kuosmanen2015d}) or the more appealing convex expectile regression (CER) (\citename{Kuosmanen2020}, \citeyear*{Kuosmanen2020}; \citename{Kuosmanen2020b}, \citeyear*{Kuosmanen2020b}) to estimate the quantile production functions. 

To reduce the dimensions of input variables in the production model, as in \citeasnoun{Lee2020}, we introduced the $\mathcal{L}_1$-norm regularization, known as Lasso (\citename{Tibshirani1996}, \citeyear*{Tibshirani1996}), for the nonparametric quantile production function.
\begin{equation}
\hat{Q}(\tau \, | \, \bx_i) = \operatorname*{arg\,min}_{f_\tau \in \Fs}\sum^{n}_{i=1}\rho_\tau (y_i - f_\tau(\bx_i)) + \lambda \left\Vert P(f_\tau) \right\Vert_1
\label{eq:pqpf1}
\end{equation}
where $\rho_\tau(t) = (\tau -1\{t \le 0\})t$ is the check function (\citename{Koenker1978}, \citeyear*{Koenker1978}), $\lambda \ge 0$ is the tuning parameter, and $\left\Vert P(f_\tau) \right\Vert_1$ denotes the Lasso regularization term.

To date, although the Lasso regularization has been predominant in the literature because of its computational advantage, several recent studies have pointed out that Lasso regularization is more of a variable screening than a variable selection procedure (e.g., \citename{Fan2010}, \citeyear*{Fan2010}; \citename{Su2017}, \citeyear*{Su2017}; \citename{Bertsimas2017}, \citeyear*{Bertsimas2017}). To effectively select the subset, another regularization $\mathcal{L}_0$ has been applied and proven to outperform the Lasso regularization, especially in a high signal-to-noise (SNR) regime (see, e.g., \citename{Bertsimas2016}, \citeyear*{Bertsimas2016}; \citename{Hastie2020}, \citeyear*{Hastie2020}). Compared with Lasso regularization, which uniformly shrinks the variables, $\mathcal{L}_0$-norm regularization sparsifies the variables only with the desired shrinking (\citename{Bertsimas2020b}, \citeyear*{Bertsimas2020b}).

We therefore further considered the $\mathcal{L}_0$-norm regularization for the nonparametric quantile production function. The $\mathcal{L}_0$-norm regularization formulation
\begin{alignat}{2}
\label{eq:pqpf2}
\hat{Q}(\tau \, | \, \bx_i) & = \operatorname*{arg\,min}_{f_\tau \in \Fs}\sum^{n}_{i=1}\rho_\tau (y_i - f_\tau(\bx_i)) &{}& \\ 
\mbox{\textit{s.t.}}\quad
& \left\Vert P(f_\tau) \right\Vert_0 \le k \notag
\end{alignat}
where $k$ represents the subset size and $\left\Vert P(f_\tau) \right\Vert_0$ denotes the $\mathcal{L}_0$-pseudonorm regularization term. Alternatively, one could apply a Lagrangian form similar to that of an Lasso regularization, that is, $\hat{Q}(\tau \, | \, \bx_i) = \operatorname*{arg\,min}_{f_\tau \in \Fs}\sum^{n}_{i=1}\rho_\tau (y_i - f_\tau(\bx_i)) + \lambda \left\Vert P(f_\tau) \right\Vert_0$. In this paper, we used an explicit formulation to ensure that the size of the support remains bounded to decrease the computational cost. To estimate the penalized quantile production models \eqref{eq:pqpf1} and \eqref{eq:pqpf2}, we propose two new penalized CQR approaches by imposing the $\mathcal{L}_1$- and $\mathcal{L}_0$-norm regularization procedures and elaborate on both methods in the following sections.

\subsection{\texorpdfstring{$\mathcal{L}_1$}{}-penalized convex quantile regression}\label{sec:l1cqr}

While the CQR estimates the full set of variables, the $\mathcal{L}_1$-penalized CQR ($\mathcal{L}_1$-CQR) aims to retain the true variables and estimate the subset selection by shrinking some coefficients and setting the others to 0. Given the prespecified tuning parameter $\lambda$, we extended the CQR problem to estimate the $\mathcal{L}_1$-CQR as follows:
\begin{alignat}{2}
 \underset{\mathbf{\alpha},\mathbf{\bbeta },{{\mathbf{\varepsilon }}^{\text{+}}},{{\mathbf{\varepsilon }}^{-}}}{\mathop{\min }}&\,\tau \sum\limits_{i=1}^{n}{\varepsilon _{i}^{+}}+(1-\tau )\sum\limits_{i=1}^{n}{\varepsilon _{i}^{-}} + \lambda \left\Vert \bbeta_i \right\Vert_1  &{}& \label{eq:pcqr}\\
\mbox{\textit{s.t.}}\quad
& y_i=\mathbf{\alpha}_i+ \bbeta_i^{'}\bx_i+\varepsilon^{+}_i - \varepsilon^{-}_i &\quad& \forall i \notag\\
& \mathbf{\alpha}_i+\bbeta_{i}^{'}{{\bx}_{i}}\le \mathbf{\alpha}_h+\bbeta _h^{'}\bx_i  &{}& \forall i,h  \notag\\
& \bbeta_i\ge \bzero &{}& \forall i  \notag\\
& \varepsilon _i^{+}\ge 0,\ \varepsilon_i^{-} \ge 0 &{}& \forall i \notag
\end{alignat}

The $\mathcal{L}_1$-CQR problem differs from the CQR problem in that the objective function adds an extra $\mathcal{L}_1$-norm regularization to $\left\Vert \bbeta_i \right\Vert_1 $, where $\left\Vert \bbeta_i \right\Vert_1 = \sum^d_j\sum^n_i\left| \beta_{j, i} \right|$. Similar to the Lasso regression, we can recover the original CQR problem as $\lambda \xrightarrow{} 0$ and remove the corresponding input variables with $\beta_{j, i}=0$ for all observations as $\lambda \xrightarrow{} \infty$. The first constraint in Problem \eqref{eq:pcqr} can be interpreted as a multivariate regression equation. The second constraint, i.e., the system of Afriat inequality, imposes the convexity. The third constraint imposes the monotonicity, and the last is the sign constraint of the error terms.

Similar to the connection between CQR and CER problems, the $\mathcal{L}_1$-penalized CER ($\mathcal{L}_1$-CER) approach is a more compelling alternative to the $\mathcal{L}_1$-CQR problem due to the nonuniqueness of linear programming problem \eqref{eq:pcqr}. Compared with the $\mathcal{L}_1$-CQR problem, the $\mathcal{L}_1$-CER problem uses an alternative quadratic objection function to ensure a unique solution to the following problem. 
\begin{alignat}{2}
 \underset{\mathbf{\alpha},\mathbf{\bbeta},{{\mathbf{\varepsilon }}^{\text{+}}},{\mathbf{\varepsilon }}^{-}}{\mathop{\min}}&\,\tilde{\tau} \sum\limits_{i=1}^n(\varepsilon _i^{+})^2+(1-\tilde{\tau} )\sum\limits_{i=1}^n(\varepsilon_i^{-})^2 + \lambda \left\Vert \bbeta_i \right\Vert_1  &{}& \label{eq:pcer}
\end{alignat}
where the $\mathcal{L}_1$-CER problem is subject to the same constraints as those of the $\mathcal{L}_1$-CQR problem. The expectile $\tilde{\tau}$ in \eqref{eq:pcer} and quantile $\tau$ in \eqref{eq:pcqr} can be converted to each other; see more discussions in \citeasnoun{Kuosmanen2020}. In practice, the $\mathcal{L}_1$-CQR problem \eqref{eq:pcqr} and $\mathcal{L}_1$-CER \eqref{eq:pcer} problem can be solved directly by commercial off-the-shelf solvers (e.g., Cplex, Mosek, and Gurobi) or standalone algorithms (e.g., \citename{Mazumder2019}, \citeyear*{Mazumder2019}; \citename{Lin2020} \citeyear*{Lin2020}). 

We notice that \citeasnoun{Lee2020} propose a similar Lasso-penalized DEA model (i.e., Lasso-SCNLS), which is, in essence, a special variant of the $\mathcal{L}_1$-CER model, i.e., expectile $\tilde{\tau} = 0.5$. However, in their study, it is unclear whether this Lasso-style penalty would necessarily reduce the dimensionality (set all $\beta_{k,i} = 0$ for some $k$) (\citename{Chen2021}, \citeyear*{Chen2021}). Although \citeasnoun{Lee2020} present some simulations that suggest improvement in the mean squared error (MSE) compared with that obtained from other DEA-related approaches, they fail to justify with convincing evidence whether this Lasso-style penalty effectively helps to reduce dimensionality. However, according to this study, the Lasso-SCNLS approach could eliminate the true irrelevant variables to some extent. Yet, the complete evidences (e.g., prediction error metric or accuracy metric) to support this argument would warrant further research.

Note further that the quadratic objective function in the $\mathcal{L}_1$-CER problem can ensure a unique estimated quantile function ($\hat{Q}$) but cannot guarantee the uniqueness of the subgradients (i.e., the estimated coefficients $\hat{\beta}_{i,j}$). To obtain a unique minimizer, denoted by ($\hat{Q}$, $\hat{\beta}_{i,j}$) in Problem \eqref{eq:pcqr} or \eqref{eq:pcer}, one could apply an alternative $\mathcal{L}_1$-norm squared regularization to the subgradients for the original CQR problem. Note that in this case, the convexity of the objective function becomes stronger, and both estimated quantile functions and estimated subgradients are then unique, a detailed proof of which would be an interesting topic for future research.

\subsection{\texorpdfstring{$\mathcal{L}_0$}{}-penalized convex quantile regression}\label{sec:l0cqr}

Considering the potential advantages of $\mathcal{L}_0$-norm regularization in subset selection, we extended the $\mathcal{L}_1$-CQR problem to the cardinality penalized CQR problem. The proposed $\mathcal{L}_0$-penalized CQR approach ($\mathcal{L}_0$-CQR) is defined as
\begin{alignat}{2}
 \underset{\mathbf{\alpha},\mathbf{\bbeta},{\mathbf{\varepsilon }^{\text{+}}},{\mathbf{\varepsilon}^{-}}}{\mathop{\min }}&\,\tau \sum\limits_{i=1}^{n}{\varepsilon _{i}^{+}}+(1-\tau )\sum\limits_{i=1}^{n}{\varepsilon _{i}^{-}} &{}&  \label{eq:pcqr2}\\ 
\mbox{\textit{s.t.}}\quad
& y_i=\mathbf{\alpha}_i+ \bbeta_i^{'}\bx_i+\varepsilon^{+}_i - \varepsilon^{-}_i &\quad& \forall i \notag\\
& \mathbf{\alpha}_i+\bbeta_{i}^{'}{{\bx}_{i}}\le \mathbf{\alpha}_h+\bbeta _h^{'}\bx_i  &{}& \forall i,h \notag \\
& \bbeta_i\ge \bzero &{}& \forall i \notag \\
& \varepsilon _i^{+}\ge 0,\ \varepsilon_i^{-} \ge 0 &{}& \forall i \notag \\
& \left\Vert \bbeta_i \right\Vert_0 \le k   &{}& \forall i \notag
\end{alignat}
where $\left\Vert \bbeta_i \right\Vert_0 := \left|j: \beta_{j, i} \neq 0 \right|, \forall j \in [d]$, denotes the $\mathcal{L}_0$ pseudonorm, indicating the number of nonzero coefficients of $\bbeta_i$. Accordingly, the last constraint in Problem \eqref{eq:pcqr2} limits the number of input variables to at most $k$. If $k \ge d$, the optimal solution to the $\mathcal{L}_0$-CQR problem \eqref{eq:pcqr2} is equivalent to the optimal solution to the original CQR problem because the cardinality constraint is not binding; if $k < d$, there are exactly $k$ input variables that contribute to the estimated quantile function. 

Although Problem \eqref{eq:pcqr2} is theoretically appealing, it is an NP-hard optimization problem (\citename{Natarajan1995}, \citeyear*{Natarajan1995}; \citename{Bertsimas2016}, \citeyear*{Bertsimas2016}). To solve Problem \eqref{eq:pcqr2} in practice, we reformulated it into a mixed integer linear programming (MILP) problem. We defined the binary variable $\bz$ to determine what the hyperplanes denoted by the pair of coefficients ($\alpha_i$, $\bbeta_i$) look like. Given an a priori-specified $k$, the MILP formulation for the $\mathcal{L}_0$-CQR problem is
\begin{alignat}{2}
 \underset{\mathbf{\alpha},\mathbf{\bbeta},{{\mathbf{\varepsilon }}^{\text{+}}},{\mathbf{\varepsilon }}^{-}, \bz}{\mathop{\min}}&\,\tau \sum\limits_{i=1}^{n}{\varepsilon _{i}^{+}}+(1-\tau )\sum\limits_{i=1}^{n}{\varepsilon _{i}^{-}} &{}& \label{eq:pcqr3}\\ 
\mbox{\textit{s.t.}}\quad
& y_i=\mathbf{\alpha}_i+ \bbeta_i^{'}\bx_i+\varepsilon^{+}_i - \varepsilon^{-}_i &\quad& \forall i \notag \\
& \mathbf{\alpha}_i+\bbeta_{i}^{'}{{\bx}_{i}}\le \mathbf{\alpha}_h+\bbeta _h^{'}\bx_i  &{}& \forall i,h \notag\\
& \bbeta_i\ge \bzero &{}& \forall i \notag\\
& \varepsilon _i^{+}\ge 0,\ \varepsilon_i^{-} \ge 0 &{}& \forall i \notag\\
& \left| \bbeta_{i} \right| \le M\bz &{}& \forall i \notag\\
& \bz \in \{0, 1 \}^d   \notag\\
& \sum^{d}_{j=1} z_j \le k \notag
\end{alignat}
where $M$ is a predetermined positive number that can be either estimated from the variants of the CQR problem (c.f., \citename{Bertsimas2020}, \citeyear*{Bertsimas2020}) or enumerated from a range of candidates. Note that a smaller $M$ could make the solution too constrained and that a larger $M$ may lead to overfitting. For the $\mathcal{L}_0$-CER problem, we can replace the objective function of \eqref{eq:pcqr3} with the quadratic objective function $\tilde{\tau} \sum\limits_{i=1}^n(\varepsilon _i^{+})^2+(1-\tilde{\tau} )\sum\limits_{i=1}^n(\varepsilon_i^{-})^2$ to ensure the uniqueness, which is now a mixed integer quadratic programming (MIQP) problem. We then employed a mixed integer optimization solver (e.g., Cplex and Gurobi) to solve the $\mathcal{L}_0$-CQR/CER problem.

\begin{Theorem}
Problems \eqref{eq:pcqr2} and \eqref{eq:pcqr3} are equivalent.
\label{theorem1}
\end{Theorem}
\begin{proof}
Let $\bbeta_i^*$ be the optimal solution to Problem \eqref{eq:pcqr2} that satisfies $\left\Vert \bbeta_i^* \right\Vert_\infty \le M$. After introducing the binary variable $\bz$, we can obtain a linear constraint on $\bz$ by imposing the cardinality constraint on $\bbeta_i$, $\bz:=\{ z \in \{0, 1\}^d: \be^{'}z \le k \}$, where all components in vector $\be$ are equal to one. Furthermore, $M\bz=M\sum^d_jz_j=M\sum^d_{j:z_j=1}z_j=M\bz_{z_j} $. Thus, $\left| \bbeta_i \right| \le M$ if $z_j=1$; otherwise, $\bbeta_i=0$. The last constraint of Problem \eqref{eq:pcqr2} can be equivalently expressed as the last three constraints of Problem \eqref{eq:pcqr3}. 
\end{proof}

Apart from the computational convenience, the explicit constraint formulation \eqref{eq:pcqr3} provides a link to the $\mathcal{L}_1$-CQR problem \eqref{eq:pcqr}. We described here the relation between the $\mathcal{L}_1$-CQR and $\mathcal{L}_0$-CQR problems. Let Conv($A$) denote the convex hull of the set of $A$. Similarly to \citeasnoun{Kuosmanen2008}, we set the $\mathcal{L}_0$-CQR problem to the function $\hat{Q}: \real^d \xrightarrow{} \real$:
\begin{alignat}{2}
    \hat{Q}(\tau \, | \, \bx):=\sup\Big\{\sum^n_{i=1}a_i\hat{y}_{\tau,i} : \sum^n_{i=1}a_i=1, \sum^n_{i=1}a_i\bx_i=\bx, a_i \ge 0, \notag \\
                                \left| \bbeta_i \right| \le M\bz, \bz \in \{0, 1\}, \sum^{d}_{j=1} z_j \le k \Big\}
     \label{eq:rep_l0}
\end{alignat}
where $a_i$ is the weight assigned to observation $i$ and $\sup(\emptyset) = -\infty$. $\hat{Q}$ is thus well defined and finite on Conv($\bx$). Using these settings, Theorem~\ref{theorem2} shows that the optimum objective value of the $\mathcal{L}_0$-CQR problem is low-bounded by that of the $\mathcal{L}_1$-CQR problem.

\begin{Theorem}
The optimal solution to Problem \eqref{eq:pcqr} is greater than or equal to that of Problem \eqref{eq:pcqr3}.
\label{theorem2}
\end{Theorem}
\begin{proof}
Similar to \citeasnoun{Bertsimas2016}, we first reformulated Conv($\bx$) \eqref{eq:rep_l0} into a relaxation form:
\begin{alignat}{2}
    \hat{Q}_1(\tau \, | \, \bx):&=\sup\Big\{\sum^n_{i=1}a_i\hat{y}_{\tau,i}: \sum^n_{i=1}a_i=1, \sum^n_{i=1}a_i\bx_i=\bx, a_i \ge 0, 
                                 \left| \bbeta_i \right| \le M\bz, \bz \in \{0, 1\}, \sum^{d}_{j=1} z_j \le k \Big\} \notag \\
                      &= \sup\Big\{\sum^n_{i=1}a_i\hat{y}_{\tau,i}: \sum^n_{i=1}a_i=1, \sum^n_{i=1}a_i\bx_i=\bx, a_i \ge 0, 
                                 \left\Vert \bbeta_i \right\Vert_\infty \le M, \left\Vert \bbeta_i \right\Vert_1 \le Mk \Big\}
   \label{eq:rep2_l0}
\end{alignat}
We then rewrote Problem \eqref{eq:pcqr} in constrained form as the convex hull:
\begin{alignat}{2}
    \hat{Q}_2(\tau \, | \, \bx):= \sup\Big\{\sum^n_{i=1}a_i\hat{y}_{\tau,i}: \sum^n_{i=1}a_i=1, \sum^n_{i=1}a_i\bx_i=\bx, a_i \ge 0, 
                                 \left\Vert \bbeta_i \right\Vert_1 \le Mk \Big\} 
    \label{eq:rep_l1}
\end{alignat}
It is straightforward to show the link between $\mathcal{L}_1$-CQR and $\mathcal{L}_0$-CQR: $\hat{Q}_1 \subseteq \hat{Q}_2$. We thus find the inequality below: $\hat{Q}_1 \ge \hat{Q}_2$. That is, Problem \eqref{eq:pcqr3} provides lower bounds for Problem \eqref{eq:pcqr}. Note that we aim to find the minimum objective function values of the $\mathcal{L}_1$-CQR and $\mathcal{L}_0$-CQR problems. 
\end{proof}

In practice, there is one concern rooted in the original CQR problem when using the $\mathcal{L}_1$-CQR and $\mathcal{L}_0$-CQR approaches. Estimating a larger sample becomes excessively expensive in both the CQR and CER approaches due to the \textit{O}($n^2$) linear constraints (\citename{Lee2013}, \citeyear*{Lee2013}; \citename{Mazumder2019}, \citeyear*{Mazumder2019}). For example, if the data samples contain 300 observations, the total number of linear constraints amounts to 90000. Undoubtedly, the computational burden becomes a barrier to the application of the $\mathcal{L}_1$-CQR and $\mathcal{L}_0$-CQR approaches under larger sample sizes and calls for a more efficient algorithm. 

Recent works are showing several promising algorithms from the perspective of computational burden (see, e.g., \citename{Lee2013}, \citeyear*{Lee2013}; \citename{Balazs2015},\citeyear*{Balazs2015}; \citename{Mazumder2019}, \citeyear*{Mazumder2019}; \citename{Bertsimas2020}, \citeyear*{Bertsimas2020}; \citename{Lin2020} \citeyear*{Lin2020}). It is worth highlighting two of them here, namely, the CNLS-G algorithm proposed by \citeasnoun{Lee2013} and the cutting-plane algorithm proposed by \citeasnoun{Balazs2015} and extended by \citeasnoun{Bertsimas2020}. Although both algorithms use the relaxed Afriat inequality constraint set and iteratively introduce new ones as necessary, the cutting-plane algorithm is more efficient to solve the CQR and CER problems due to its fast identification of violating constraints. We thus improved on the base of the cutting-plane algorithm in \citeasnoun{Bertsimas2020} to solve Problems \eqref{eq:pcqr} and \eqref{eq:pcqr3} and referred to it as the CNLS-adapted cutting-plane (CNLS-A) algorithm. In the Appendix \ref{sec:anlsa}, we introduce the CNLS-A algorithm in detail and perform a computational test to compare the performances of the CNLS-G and CNLS-A algorithms in solving the CQR and CER problems.


\section{Monte Carlo study}\label{sec:mc}
\setcounter{equation}{0}

\subsection{Setup}

We present an MC study to compare the finite sample performances in reducing the dimensionality between the $\mathcal{L}_1$- and $\mathcal{L}_0$-penalized CQR/CER models. Following \citeasnoun{Lee2013}, we considered an additive Cobb--Douglas production function with $d$ input variables and one output variable. Given $n$ (number of observations), $d$ (number of input variables), $k$ (number of true input variables), $\rho$ (the SNR level), and $\tau$ (the quantile level), our experimental process was as illustrated below:
\begin{enumerate}[label= \arabic*), leftmargin=2\parindent]
    \setlength{\itemsep}{1pt}
    \setlength{\parskip}{0pt}
    \setlength{\parsep}{0pt}
    \item we defined the following parameters to generate the artificial data: $n \in \{100, 500\}$, $d \in \{6, 8, 10, 12\}$, $\rho \in \{0.5, 2, 10\}$, and $\tau \in \{0.1, 0.3, 0.5, 0.7, 0.9\}$;
    \item we randomly sampled a support set ($\omega$) of size $k$ from $\{1, 2, \cdots, d\}$ and considered that the sizes of the true support set ($\omega^*$) are 2 and 4, respectively;
    \item we drew the input variables $\bx_i \in \real^{n \times d}$ i.i.d. from a uniform distribution, $U[1, 10]$;
    \item the observed output $y_i = \prod_{d=1}^{D \in \omega^*}\bx^{\frac{0.8}{d}}_{d,i} + v_i$. For each noise, $v_i$ was randomly sampled from the normal distribution, $N(0, \sigma^2)$. $\sigma^2$ was calculated based on the SNR level, $\sigma^2 = var(\prod_{d=1}^{D \in \omega^*}\bx^{\frac{0.8}{d}}_{d,i}) / \rho $ (\citename{Bertsimas2016}, \citeyear*{Bertsimas2016});
    \item we conducted the $\mathcal{L}_1$-CQR, $\mathcal{L}_0$-CQR, $\mathcal{L}_1$-CER, and $\mathcal{L}_0$-CER on the data $\bx_i$, $y$, with the optimal tuning parameters; for each approach, the tuning parameter was selected by a 5-fold cross-validation procedure from a wide range of values (\citename{Bertsimas2020}, \citeyear*{Bertsimas2020}; \citename{chen2020c}, \citeyear*{chen2020c}; \citename{Chen2021a}, \citeyear*{Chen2021a});
    \item we measured the finite sample performance of each approach using the following two statistics: prediction error and accuracy;
    \item we replicated steps 1)-6) 10 times and averaged the statistics.
\end{enumerate}

To assess the finite sample performance of each method, we utilized the prediction error and accuracy statistics, which are, respectively, the level of good data fidelity and the proportion of true variables that were selected. For a penalized quantile function estimate, the prediction error statistic is defined as (\citename{Bertsimas2016}, \citeyear*{Bertsimas2016})
\begin{equation*}
\text{Prediction error} = \left\Vert \hat{Q}_i - Q^{\omega^*}_i \right\Vert^2_2 / \left\Vert Q^{\omega^*}_i \right\Vert^2_2
\end{equation*}
where $ \hat{Q}_i$ denotes the estimated conditional penalized quantile function, identified by a particular approach, and $Q^{\omega^*}_i$ indicates the true conditional penalized quantile function. Note that for the prediction error statistic, an ``in-sample" version is used in this paper, where the smaller the prediction error is, the better the performance will be. Furthermore, we reported the proportion of true variables that were selected, that is, the accuracy (\%) (\citename{Bertsimas2020}, \citeyear*{Bertsimas2020})
\begin{equation*}
\text{Accuracy} = \frac{\left|\hat{\omega}\cap\omega^*\right|}{k} \times 100\%
\end{equation*}

Regarding the tuning parameter selection, by using a 5-fold cross-validation procedure (see, e.g., \citename{Bertsimas2020}, \citeyear*{Bertsimas2020}; \citename{chen2020c}, \citeyear*{chen2020c}; \citename{Chen2021a}, \citeyear*{Chen2021a}), we calibrated the parameter $\lambda$ over a sequence of 100 values for the $\mathcal{L}_1$-CQR/CER approach. For the $\mathcal{L}_0$-norm regularization, we varied the tuning value $M$ from the multiplier set $\{0.1, 0.5, 1, 1.5, 2, 3, 4, 5\}$ for the $\mathcal{L}_0$-CQR/CER approach, and we chose the subset size $k$ from the set $\{1, \cdots, d-1\}$ (i.e., at least one input). Note that the 5-fold cross-validation technique needs to solve each $\mathcal{L}_1$/$\mathcal{L}_0$-CQR/CER problem five times for each candidate value; a faster procedure for finding the tuning parameters for a similar optimization problem can be seen in \citeasnoun{chen2020c}.

In all the experiments that follow, we resorted to Julia/JuMP (\citename{Dunning2017}, \citeyear*{Dunning2017}) to solve the $\mathcal{L}_1$- and $\mathcal{L}_0$-CQR/CER models based on the CNLS-A algorithm. We used the commercial off-the-shelf solver Gurobi (9.1.2) to solve the CNLS-A algorithm. All experiments were performed on Aalto University's high-performance computing cluster Triton with Xeon @2.8 GHz processors, 5 CPUs, and 5 GB RAM per CPU. Since the MC study is naturally parallel, we resorted to the multiple threads in Julia to speed up the computation. The skeleton codes are publicly available at \url{https://github.com/ds2010/penalizedCQR}.

\subsection{Prediction error}\label{sec:pe}

We first explored the prediction error of the $\mathcal{L}_1$- and $\mathcal{L}_0$-norm regularization approaches solved by the CNLS-A algorithm. For the comparison of the quantile and expectile approaches, the expectile $\tilde{\tau}$ was transformed into the corresponding quantile $\tau$ based on the empirical inverse quantile functions of the noise ($\varepsilon_i$).

Figs.~\ref{fig:fig1} and \ref{fig:fig2} depict how the prediction error scales as the sample sizes, input dimensions, and quantiles increase when estimating the $\mathcal{L}_1$-CQR and $\mathcal{L}_0$-CER problems with the observation $n=100$ but different sizes of the true support set (i.e., $k=2$ and $k=4$). Recall that the smaller the prediction error is, the better the performance. Combined with Figs.~\ref{fig:figb1} and \ref{fig:figb2}, $\mathcal{L}_0$-CQR shows a major advantage over $\mathcal{L}_1$-CQR in quantile production function estimation when applying the subset selection procedure. The main findings are summarized as follows.
\begin{figure}[!htbp]
    \vspace{-1em}
    \centering
    \includegraphics[width=1\textwidth]{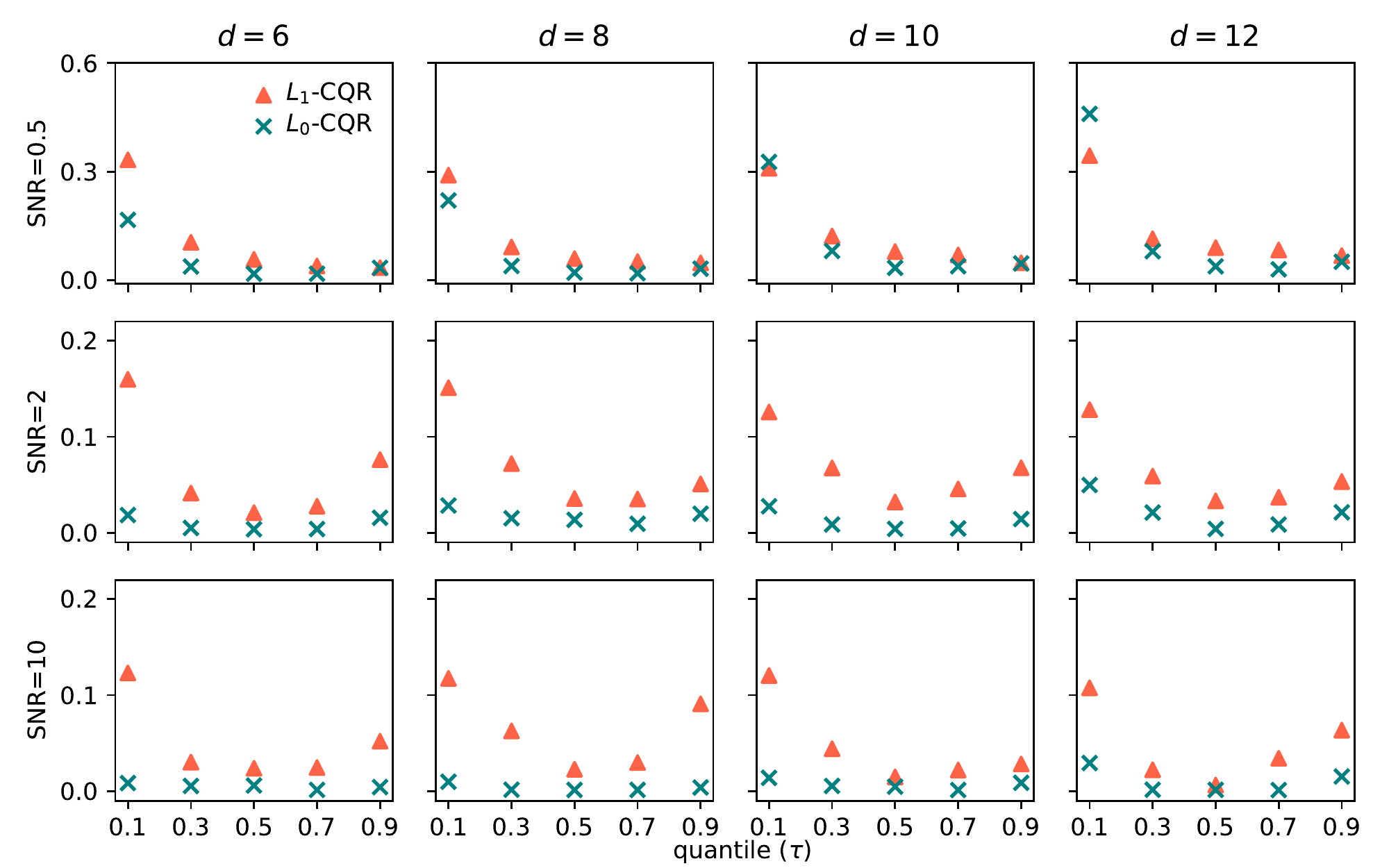}
    \caption{Prediction errors of the $\mathcal{L}_1$-CQR and $\mathcal{L}_0$-CQR approaches with $n = 100$ and $k_{\text{true}} = 2$.}
    \label{fig:fig1}
    \vspace{-1em}
\end{figure}

\begin{figure}[!htbp]
    \vspace{-1em}
    \centering
    \includegraphics[width=1\textwidth]{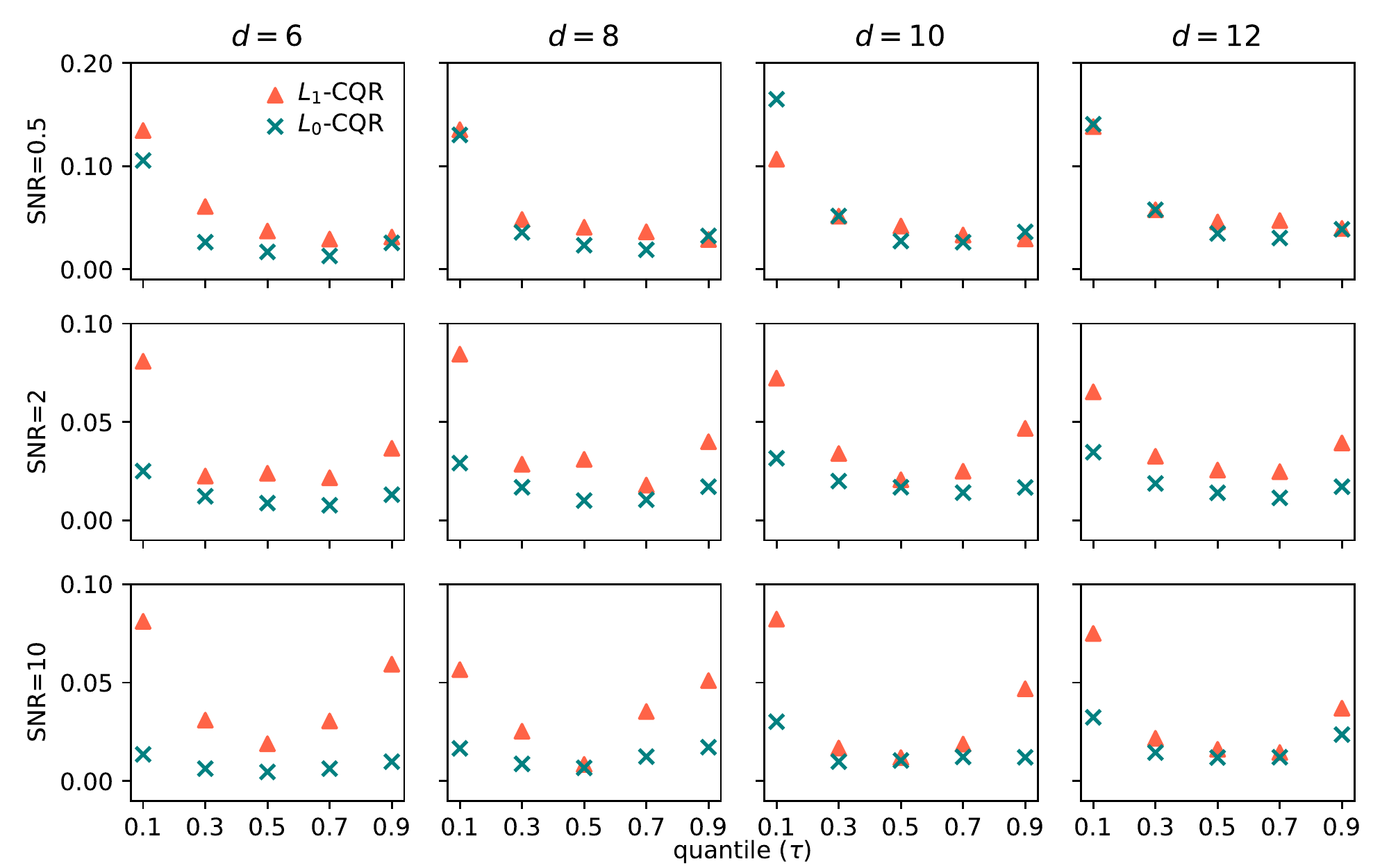}
    \caption{Prediction errors of the $\mathcal{L}_1$-CQR and $\mathcal{L}_0$-CQR approaches with $n = 100$ and $k_{\text{true}} = 4$.}
    \label{fig:fig2}
    \vspace{-1em}
\end{figure}

\begin{itemize}
\item  The $\mathcal{L}_0$-CQR approach generally outperformed the $\mathcal{L}_1$-CQR approach based on the comparison of the prediction errors, especially in a high-SNR regime. However, the $\mathcal{L}_1$-CQR performed better in certain cases in a low-SNR regime, which confirms the conclusion by \citeasnoun{Hastie2020}.
\item The quantile did not influence the prediction error systematically. As shown in Figs.~\ref{fig:fig1}, \ref{fig:fig2}, \ref{fig:figb1}, and \ref{fig:figb2}, there was a ``U-shaped" relationship between the prediction error and the quantile for both the $\mathcal{L}_1$- and $\mathcal{L}_0$-CQR approaches. The lowest bias was observed in the median estimation; the higher bias was observed in the lower-quantile or the higher-quantile estimation.
\item The $\mathcal{L}_0$-CQR and $\mathcal{L}_1$-CQR approaches tended to yield larger prediction errors as the number of dimensions increased, especially for a large true support set. \citeasnoun{Lee2020} also report similar conclusions for the Lasso-SCNLS in terms of the MSE change; that is, the MSE increases as more inputs are included. This is because the high-dimensional data space usually has a large sparsity, which undermines the prediction accuracy of both approaches.
\item The noise term obviously affected the precision of the estimated quantile function. As the SNR increased, the prediction error became smaller for both the $\mathcal{L}_1$-CQR and $\mathcal{L}_0$-CQR approaches. Furthermore, when the size of the true support set increased, the prediction error also decreased. 
\item The sample size effect was also observed in the experiments. Compared to estimating a larger sample size (Figs.~\ref{fig:figb1} and \ref{fig:figb2}), the results of Figs.~\ref{fig:fig1} and \ref{fig:fig2} suggest that a large sample would decrease the prediction error in both approaches. A similar conclusion is commonly seen in most of the existing literature (see, e.g., \citename{Lee2020}, \citeyear*{Lee2020}; \citename{Chen2021}, \citeyear*{Chen2021}). 
\end{itemize}

Because of the nonuniqueness of estimated objective function $\hat{Q}$ in the penalized quantile approaches, it is unfair to compare the performances of the penalized quantile and expectile approaches in estimating the quantile production function. In fact, the prediction errors (which are available upon request) show that neither the quantile approaches nor the expectile approaches clearly outperform the others. Within this context of subset selection, the core competitiveness of the methods relies on whether they can accurately eliminate the irrelevant variables. The accuracy metric would thus be a better measurement approach.

\subsection{Accuracy}\label{sec:acc}

To provide direct evidence of the performance of the $\mathcal{L}_1$- and $\mathcal{L}_0$-norm regularization approaches in subset variable selection, we next evaluated the accuracy statistic by counting the number of nonzero estimated $ \bbeta_i $, which is compared with the true support set $\omega^*$.  Recall that the higher the accuracy is, the better the performance.

Figs.~\ref{fig:fig3} and \ref{fig:figb3} plot the accuracy results with $n=500$ and $n=100$, respectively, but yield somewhat similar results. In most cases, the $\mathcal{L}_0$-CQR approach outperformed the $\mathcal{L}_1$-CQR approach. As expected, it was observed that the accuracy decreased in all scenarios with increasing sample size and that $\mathcal{L}_0$-CQR achieved the highest accuracy. In certain cases (e.g., low quantile with $\tau =0.1$), the accuracy of the $\mathcal{L}_1$-CQR approach was zero, indicating that the $\mathcal{L}_1$-CQR approach was unable to completely eliminate the irrelevant variables compared with the true support set. However, the $\mathcal{L}_1$-CQR approach can successfully select the true variables, at least partially. Therefore, for accuracy comparison, the $L_0$-CQR approach is a better choice for subset selection.
\begin{figure}[!htbp]
    \vspace{-1em}
    \centering
    \includegraphics[width=1\textwidth]{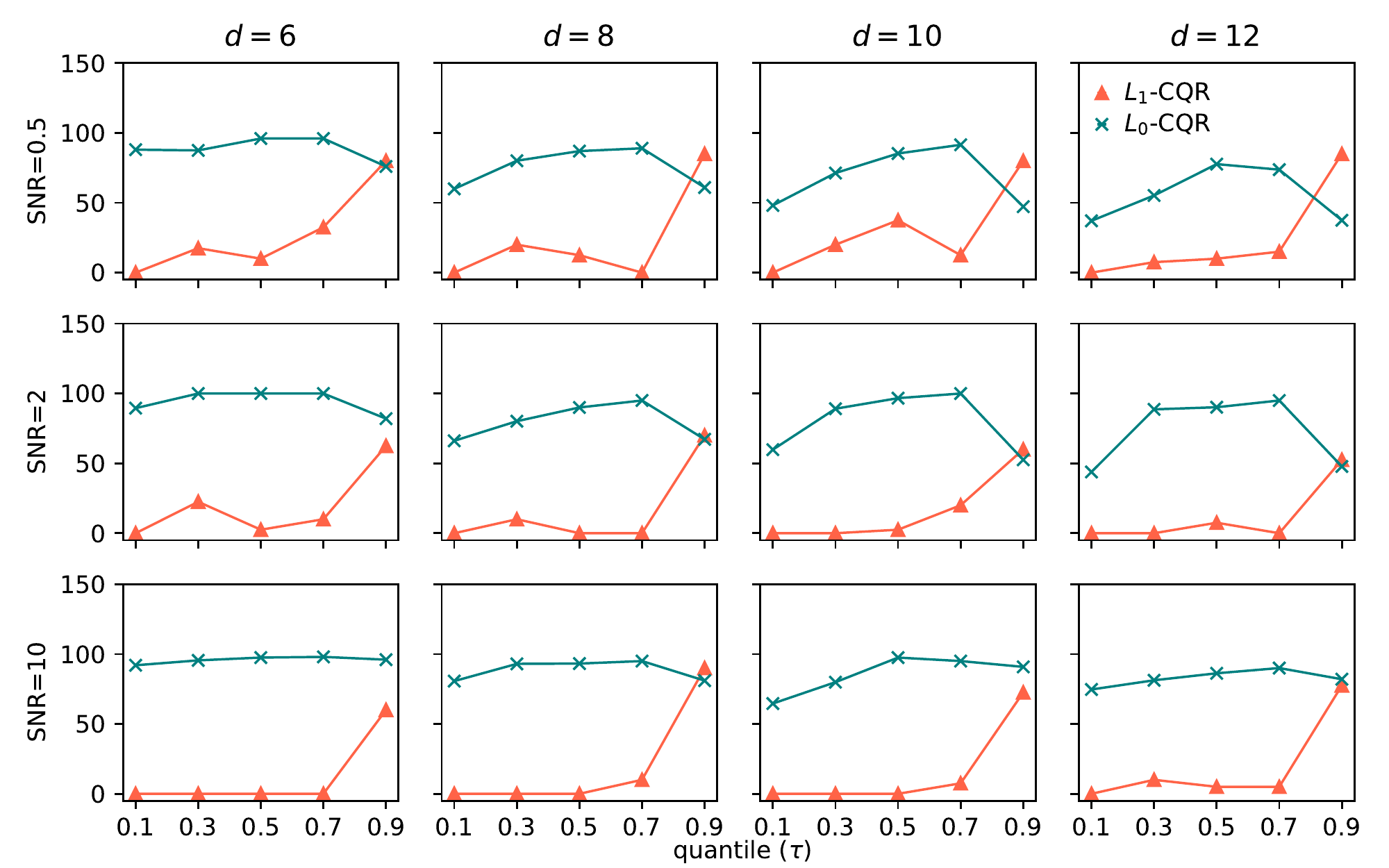}
    \caption{Accuracy of the $\mathcal{L}_1$-CQR and $\mathcal{L}_0$-CQR approaches with $n = 500$ and $k = 4$.}
    \label{fig:fig3}
    \vspace{-1em}
\end{figure}

Furthermore, similar to the prediction error comparison, several other main findings based on the accuracy benchmarking are summarized as follows:
\begin{itemize}
\item Fig.~\ref{fig:fig3} also shows no clear deterministic trend with decreasing or increasing quantile. For the $\mathcal{L}_1$-CQR approach, the lower-quantile estimation yielded poor accuracy and even could be zero. By contrast, the $\mathcal{L}_0$-CQR approach achieved a higher accuracy at the median quantile. Furthermore, the accuracy of $\mathcal{L}_0$-CQR was more than 50\% in almost all scenarios and close to 100\% in large parts. Note that the accuracy difference in each quantile estimation provides direct evidence of the superiority of the $\mathcal{L}_0$-CQR approach, especially in a high-SNR regime.
\item As more input variables were included, using both methods decreased the accuracy in terms of subset variable selection. This is due to the fact that the higher dimension makes the problem more sparse, inducing difficulty in true variable identification and selection. Although the $\mathcal{L}_1$-CQR approach achieved lower accuracy when $d=12$ with the lower SNR, $\mathcal{L}_1$-CQR generally performed better in most scenarios.
\item The noise term also affected the accuracy of subset variable selection. The larger noise amplitude in the data resulted in unstable accuracy estimation and normally yielded lower accuracy results. Furthermore, considering the results demonstrated in Figs.~\ref{fig:figb3}-\ref{fig:figb5}, the smaller size of the true support set made the subset selection easier for both approaches, as reflected by the low accuracy.
\end{itemize}

We further investigated the accuracy of the penalized quantile and expectile approaches for different input dimensions, true support set sizes, and SNR levels. Table~\ref{tab:tab1} demonstrates the simulation results with sample size $n=100$ and quantile $\tau=0.9$. As shown in Table~\ref{tab:tab1}, the expectile estimation approach generally outperformed the quantile approach. That is, the expectile approach can select the subset variable more effectively. This is because the expectile approach has a stronger convex objective function. However, the expectile approach performed worse compared to the quantile approach in certain cases, possibly due to the well-known small sample bias. Since the performances of the $\mathcal{L}_1$- and $\mathcal{L}_0$-norm regularization approaches show a nonsystematic trend as the quantile increases or decreases, the comparison for other quantiles also presents similar results to those in Table~\ref{tab:tab1}.
\begin{table}[!htbp]
\caption{Comparison of the accuracy of quantile and expectile estimation for $\mathcal{L}_0$-norm regularization.}
\label{tab:tab1}
\renewcommand\arraystretch{1.15}
\begin{center}
    \vspace{-1em}
    \begin{tabular}{lrrrrrrr}
    \toprule
    \multicolumn{1}{c}{\multirow{2}[4]{*}{$d$}} & \multicolumn{1}{c}{\multirow{2}[4]{*}{$k$}} & \multicolumn{2}{c}{$\text{SNR}=0.5$} & \multicolumn{2}{c}{$\text{SNR}=2$} & \multicolumn{2}{c}{$\text{SNR}=10$} \\
    \cmidrule{3-8}          &   & \multicolumn{1}{c}{quantile} & \multicolumn{1}{c}{expectile} & \multicolumn{1}{c}{quantile} & \multicolumn{1}{c}{expectile} & \multicolumn{1}{c}{quantile} & \multicolumn{1}{c}{expectile} \\
    \midrule
    6     & 2     & 42    & 42    & 43    & 52    & 56    & 63 \\
          & 4     & 76    & 78    & 74    & 73    & 69    & 79 \\
    8     & 2     & 39    & 34    & 39    & 39    & 46    & 39 \\
          & 4     & 54    & 51    & 53    & 61    & 47    & 63 \\
    10    & 2     & 22    & 24    & 45    & 34    & 47    & 48 \\
          & 4     & 44    & 47    & 38    & 49    & 65    & 56 \\
    12    & 2     & 24    & 21    & 39    & 25    & 37    & 40 \\
          & 4     & 29    & 47    & 31    & 46    & 44    & 46 \\
    \bottomrule
    \end{tabular}%
\end{center}
\vspace{-2em}
\end{table}%

In conclusion, the experiments seem informative enough to indicate that $\mathcal{L}_0$-CQR can outperform $\mathcal{L}_1$-CQR in terms of prediction error and accuracy. This is particularly true when the $\mathcal{L}_1$-  and $\mathcal{L}_0$-CQR associated with a large sample size or higher-noise data are estimated. The experiments also reveal that the expectile approaches generally perform better than the quantile approaches. Thus, we apply the expectile-based regularization approaches ($\mathcal{L}_1$-  and $\mathcal{L}_0$-CER) to select the subset variables and analyze the inequality in sustainable development in OECD countries.


\section{Application to SDG evaluation}\label{sec:sdg}
\setcounter{equation}{0}

In this section, we explore the effectiveness of the proposed $\mathcal{L}_1$- and $\mathcal{L}_0$-norm regularization approaches in reducing dimensionality empirically and evaluated the performance of sustainability in 35 OECD countries using the SDG indicators. 

\subsection{Data and variables}

To benchmark the degree of sustainable development for OECD countries, we estimated the quantile production function based on the panel data of 35 OECD countries during 2017, 2019, and 2020.\footnote{
Colombia and Lithuania are excluded from the list of OECD countries due to insufficient data. The data for 2018 are also excluded due to the inconsistent indicators from the same database.} 
We collected the raw data from the UN SDG Index and Dashboard Reports (\citename{Sachs2017}, \citeyear*{Sachs2017}, \citeyear*{Sachs2019}, \citeyear*{Sachs2020}), which provide a substantial range of indicators and detailed descriptions of the indicator definition, indicator measurement, and indicator source. The same SDG data have already been applied in other empirical applications of sustainability evaluation (see, e.g., \citename{Lamichhane2020}, \citeyear*{Lamichhane2020}; \citename{Huan2021}, \citeyear*{Huan2021}).

The input-output variables in this paper were selected based on three pillars of sustainability: social, economic, and environmental. In the framework of sustainable development, the 17 SDGs can be divided into these three aspects but with different classifications (\citename{Dalampira2020}, \citeyear*{Dalampira2020}). For example, environmental sustainability can be achieved through Goals 6, 13, 14, and 15; economic sustainability through Goals 7, 8, 9, 11, and 12; and social sustainability through Goals 1, 2, 3, 4, 5, 10, 16, and 17 (\citename{Costanza2016}, \citeyear*{Costanza2016}). Using this criterion combined with the UN SDG Index and Dashboard Reports, we retrieved the following variables to characterize the production of sustainability (Table~\ref{tab:ind}).
\begin{table}[!htbp]
    \renewcommand\arraystretch{1.15}
    \caption{Description of the selected indicators for the sustainability evaluation.}
    \label{tab:ind}
    \centering
    \begin{threeparttable}
    \begin{tabular}{lll}
    \toprule
    \multicolumn{1}{l}{Pillar} & \multicolumn{1}{l}{Goal} & \multicolumn{1}{l}{Indicator} \\
    \midrule
                        &   \vspace{-1em}\\  
        Economic        &  SDG8     &  ($y$) Adjusted GDP growth (\%) \\
                        &   \vspace{-0.5em} \\  
        Social          &  SDG1     &  ($x_1$) Poverty headcount ratio at US\$ 1.90/day (\%) \\
                        &  SDG3     &  ($x_2$) Life expectancy at birth (years) \\
                        &           &  ($x_3$) Subjective well-being (average ladder score) \\
                        &  SDG8     &  ($x_4$) Employment-to-population ratio (\%) \\
                        &  SDG10    &  ($x_5$) Gini coefficient adjusted for top income (1-100) \\
                        &   \vspace{-0.5em} \\  
        Environmental     &  SDG11    &  ($x_6$) Annual mean concentration of PM2.5 ($\mu$g/m$^3$) \\
                        &  SDG12    &  ($x_7$) Municipal solid waste (kg/day/capita) \\
                        &           &  ($x_8$) E-waste generated (kg/capita) \\
                        &           &  ($x_9$) Production-based SO$_2$ emissions (kg/capita) \\
                        &           &  ($x_{10}$) Nitrogen production footprint (kg/capita) \\
                        &  SDG13    &  ($x_{11}$) Energy-related CO$_2$ emissions per capita (tCO$_2$/capita) \\
                        &           &  ($x_{12}$) Effective carbon rate (EUR/tCO$_2$)  \\
    \bottomrule
    \end{tabular}%
\end{threeparttable}
\end{table}%

To integrate the sustainability evaluation (or SDG assessment) into the production function, the aggregated variable, GDP, from the economic dimension is treated as the output, and the other variables from the social and environmental dimensions are considered as the inputs. In the estimation, we used the real GDP collected from the World Bank (the same data source) instead of the adjusted GDP growth. Regarding the indicator employment-to-population ratio, we transformed it into the actual population according to its definition in the dashboard reports. Table~\ref{tab:stats} summarizes the descriptive statistics for the input and output variables. Note that one can, of course, introduce such sustainable indicators as the elderly poverty rate, inequality of income, and freshwater withdrawal described in \citeasnoun{Lamichhane2020}, \citeasnoun{Singpai}, and the references therein, but should mind the computational burden in the nonparametric model.

\subsection{Subset selection and estimates}

Since the expectile approaches generally outperform the quantile approaches in terms of the accuracy comparison and can ensure uniqueness of the estimated quantile functions, we applied the $\mathcal{L}_1$- and $\mathcal{L}_0$-CER approaches elaborated in Section \ref{sec:meth} to calculate the averaged estimates and examine the effectiveness in subset variable selection. To obtain a more complete picture of the distribution characteristics of sustainability, we considered 10 expectile functions with $\tilde{\tau} \in  \{0.05, 0.15, \cdots, 0.85, 0.95\}$ by solving the $\mathcal{L}_1$-CER and $\mathcal{L}_0$-CER problems. The corresponding quantile $\tau$ was reported to connect with the usual quantile estimation by counting the number of negative residuals $\hat{\varepsilon}_i^-$ that are greater than zero (\citename{Efron1991}, \citeyear*{Efron1991}).

For the SDG application, we also used the 5-fold cross-validation procedure to determine the optimal tuning parameters $\lambda$, $k$, and $M$. Specifically, we calibrated $\lambda$ over 100 values ranging from 0.1 to 3 for the $\mathcal{L}_1$-CER approach. For the $\mathcal{L}_0$-CER approach, the subset size $k$ and constant $M$ were chosen from a set of integers ranging from 1 to 11 and the set $\{0.1, 0.5, 0.8, 1, 1.5, 1.8, 2, 2.5, 3, 5\}$, respectively. Table~\ref{tab:esti} demonstrates the averaged estimates of each input variable on sustainability for the selected expectiles (i.e., 5\textsuperscript{th}, 35\textsuperscript{th}, 65\textsuperscript{th}, and 95\textsuperscript{th}) in $\mathcal{L}_1$-CER and $\mathcal{L}_0$-CER estimation. 

As shown in Table \ref{tab:esti}, the $\mathcal{L}_0$-norm regularization is a better choice to select the subset variables, which further confirms the conclusion of the MC study. The blanks in Table \ref{tab:esti} denote that the corresponding variable is eliminated by the $\mathcal{L}_1$-norm or $\mathcal{L}_0$-norm regularization. The number of eliminated irrelevant variables in the $\mathcal{L}_0$-CER estimation is larger than that in the $\mathcal{L}_1$-CER estimation for all expectiles. For example, in the 95\textsuperscript{th} expectile (i.e., 63\textsuperscript{rd} and 71\textsuperscript{st} quantiles for the $\mathcal{L}_0$ penalty and $\mathcal{L}_1$ penalty, respectively), $\mathcal{L}_0$-CER selects four principal input variables, whereas $\mathcal{L}_1$-CER estimates eight input variables. Note that the subset size of actual variables used in the $\mathcal{L}_0$-CER estimation is less than or equal to the optimal $\hat{k}$.
\begin{table}[!htbp]
  \renewcommand\arraystretch{1.15}
  \centering
  \caption{Average estimates of the subset variables for the selected expectiles.}
    \begin{footnotesize}
    \begin{tabular}{lrrrrrrrrrrr}
    \toprule
          & \multicolumn{1}{c}{$\mathcal{L}_0$-CER} & \multicolumn{1}{c}{$\mathcal{L}_1$-CER} &       & \multicolumn{1}{c}{$\mathcal{L}_0$-CER} & \multicolumn{1}{c}{$\mathcal{L}_1$-CER} &       & \multicolumn{1}{c}{$\mathcal{L}_0$-CER} & \multicolumn{1}{c}{$\mathcal{L}_1$-CER} &       & \multicolumn{1}{c}{$\mathcal{L}_0$-CER} & \multicolumn{1}{c}{$\mathcal{L}_1$-CER} \\
         \cmidrule{2-3}\cmidrule{5-6}\cmidrule{8-9}\cmidrule{11-12}    E ($\tilde{\tau}$) & \multicolumn{2}{c}{0.05} &       & \multicolumn{2}{c}{0.35} &       & \multicolumn{2}{c}{0.65} &       & \multicolumn{2}{c}{0.95} \\
        \cmidrule{2-3}\cmidrule{5-6}\cmidrule{8-9}\cmidrule{11-12}    Q ($\tau$) & 0.14  & 0.16  &       & 0.29  & 0.36  &       & 0.46  & 0.42  &       & 0.63  & 0.71 \\
    \midrule
    $\hat{\beta}_{1}$ &        &        &       &       &       &       &       &       &       &       &  \\
    $\hat{\beta}_{2}$ & 0.084  & 0.0006 &       & 0.104  & 0.0011 &       &       & 0.0007 &       & 0.046  & 0.0008 \\
    $\hat{\beta}_{3}$ & 0.519  &        &       & 0.370  &       &       &       &       &       &       &  \\
    $\hat{\beta}_{4}$ & 0.213  & 0.1039 &       & 0.251  & 0.1197 &       & 0.168  & 0.1203 &       & 0.181  & 0.1241 \\
    $\hat{\beta}_{5}$ &        & 0.0001 &       &       & 0.0010 &       &       & 0.0001 &       &       & 0.0002 \\
    $\hat{\beta}_{6}$ &        & 0.0003 &       &       & 0.0014 &       &       & 0.0007 &       &       & 0.0005 \\
    $\hat{\beta}_{7}$ &        &        &       & 0.766  &       &       & 0.417  &       &       &       &  \\
    $\hat{\beta}_{8}$ & 0.145  & 0.0008 &       & 0.000  & 0.0031 &       & 0.000  & 0.0009 &       & 0.079  & 0.0011 \\
    $\hat{\beta}_{9}$ &        & 0.0033 &       &       & 0.0136 &       &       & 0.0061 &       &       & 0.0086 \\
    $\hat{\beta}_{10}$ &       & 0.0122 &       &       & 0.0266 &       &       & 0.0184 &       &       & 0.0246 \\
    $\hat{\beta}_{11}$ &       &        &       &       & 0.0052 &       & 0.172  &       &       & 0.161  &  \\
    $\hat{\beta}_{12}$ &       & 0.0036 &       &       & 0.0054 &       &       & 0.0012 &       &       & 0.0010 \\
    $\hat{\alpha}$     & -10.249  & -0.2700 &       & -10.300  & -0.8202 &       & -1.242  & -0.2407 &   & -4.356  & -0.3166 \\
    $\hat{k}$          & 7     &  --   &      & 7     &  --    &       & 3     &  --     &       & 6     & -- \\
    $\hat{M}$          & 2.5   &  --   &      & 2.5   &  --    &       & 1.5   &  --     &       & 1.5   & -- \\
    $\hat{\lambda}$    & --    & 0.95  &      & --    & 0.33   &       & --    & 2.41    &       & --    & 0.31 \\
    \bottomrule
    \end{tabular}%
    \label{tab:esti}%
    \end{footnotesize}
\end{table}%

The best subset variables differ across different expectiles (quantiles). It can be seen that, for example, variable $x_2$, life expectancy, is eliminated at the 65\textsuperscript{th} expectile in $\mathcal{L}_0$-CER but not at the 5\textsuperscript{th}, 35\textsuperscript{th}, or 95\textsuperscript{th} expectiles. Note that the CER approach, a novel data-driven approach, estimates the production function locally and is robust to heterogeneity (\citename{Kuosmanen2020b}, \citeyear*{Kuosmanen2020b}; \citename{Dai2020}, \citeyear*{Dai2020}). The approaches proposed in this paper inherit the good properties from the CER. Utilizing the full information of each observation, the $\mathcal{L}_1$-/$\mathcal{L}_0$-CER approach estimates the expectile production function based on the actual level of sustainability, taking both inefficiency and noise into account, but not the full production frontier. The subset selection is thus not exactly the same in each expectile estimation. 

The average estimates of subset variables reveal the heterogeneous effects on the expected sustainability from the expectile estimation. Note that the estimated coefficients $\hat{\beta}$ differ across observations. In the $\mathcal{L}_0$-CER estimation, we can see that the impact of the employment population in the lower 5\textsuperscript{th} and 35\textsuperscript{th} expectile countries are higher than those in other expectile countries. From the policy-recommendation perspective, those countries located at the lower expectile can take further actions to promote employment and even sustainable development. The heterogeneous policies could be proposed according to the average estimates. 

Table \ref{tab:esti} shows the strong evidence of economies of scale (i.e., increasing returns to scale) for OECD countries. We find that the estimated intercepts $\hat{\alpha}$ are systematically less than zero for all expectiles, exhibiting increasing returns to scale. In terms of each expectile, all countries are operating at increasing returns to scale in the lower expectiles (e.g., 5\textsuperscript{th} and 15\textsuperscript{th} expectiles), while few countries are operating at decreasing returns to scale in the higher expectiles (e.g., only three units with positive $\hat{\alpha}$ in the 95\textsuperscript{th} expectile). In the context of sustainability production, increasing returns suggest that the OECD countries could improve sustainability through further expansion. The OECD countries, especially for the smaller counties, for example, could increase inputs in the social aspect (e.g., increasing the employment population) and decrease pollutant emissions in the environmental dimension to accelerate sustainable development. 

\subsection{Performance benchmarking}

While the average estimates of subset variables are very interesting and relevant as such, another main purpose of this paper is to benchmark the performance of sustainable development for OECD countries. To this end, Table \ref{tab:perf} reports the productive performance of OECD countries relative to the quantile frontier based on the $\mathcal{L}_0$-CER estimation. For the sake of illustration, following \citeasnoun{Kuosmanen2020}, Table \ref{tab:perf} presents the full list of OECD countries under four groups: EU-15, EU transition economies, European Free Trade Association (EFTA), and Non-European OECD. Note that the higher the expectile is, the better the relative performance.

Recall that the $\mathcal{L}_0$-CER approach estimates the quantile (expectile) production function instead of the full frontier. This would suggest that the quantile production function estimation can serve as the benchmarking for structure analysis. There is thus a good connection between the quantile benchmarking and metafrontier benchmarking. One could build a quantile metafrontier to analyze a group of units at the same quantile level (see \citename{Lai2018}, \citeyear*{Lai2018}). In the present context, the quantile output technical efficiency can be measured by $\text{TE}(\tau)= y /\hat{Q}(\tau \, | \, \bx)$. However, for the purpose of performance benchmarking, one can directly compare the relative locations of the quantiles (e.g., \citename{Kuosmanen2020}, \citeyear*{Kuosmanen2020}; this paper) instead of the quantile-based technical efficiency (e.g., \citename{Behr2010}, \citeyear*{Behr2010}; \citename{Lai2018}, \citeyear*{Lai2018}). 
\begin{table}[!htbp] 
    \renewcommand\arraystretch{1.15}
    \caption{Productive performances of countries relative to the quantile frontiers.}
    \label{tab:perf}%
    \centering
    \begin{threeparttable}
    \begin{tabular}{lcccrccc}
    \toprule
          & 2017  & 2019  & 2020  &       & 2017  & 2019  & 2020 \\
    \midrule
    \multicolumn{4}{l}{\textbf{EU-15}} & \multicolumn{4}{l}{\textbf{EU transition economies}} \\
    Luxembourg & 10    & 10    & 10    & \multicolumn{1}{l}{Estonia} & 10    & (6, 7) & 10 \\
    United Kingdom & 1     & 1     & (1, 2) & \multicolumn{1}{l}{Slovenia} & 10    & 10    & 10 \\
    Denmark & 10    & (9, 10) & 10    & \multicolumn{1}{l}{Poland} & (5, 6) & (6, 7) & (2, 3) \\
    Sweden & (9, 10) & (9, 10) & (9, 10) & \multicolumn{1}{l}{Hungary} & (6, 7) & (1, 2) & (5, 6) \\
    Germany & (2, 3) & (1, 2) & 1     & \multicolumn{1}{l}{Latvia} & 10    & 10    & 10 \\
    Finland & (8, 9 & (8, 9) & (9, 10) & \multicolumn{1}{l}{Czech Republic} & (9, 10) & (9, 10) & (8, 9) \\
    Belgium & 10    & (9, 10) & (9, 10) & \multicolumn{1}{l}{Slovakia} & (9, 10) & (9, 10) & (7, 8) \\
    Netherlands & (6, 7) & (6, 7) & (6, 7) &       &       &       &  \\
    France & (3, 4) & (4, 5) & (7, 8) & \multicolumn{4}{l}{\textbf{Non-European OECD}} \\
    Italy & (4, 5) & (4, 5) & (3, 4) & \multicolumn{1}{l}{Canada} & 1     & 1     & \multicolumn{1}{l}{(1, 2)} \\
    Austria & (7, 8) & (8, 9) & (7, 8) & \multicolumn{1}{l}{Israel} & (6, 7) & (6, 7) & (6, 7) \\
    Greece & (6, 7) & (7, 8) & (7, 8) & \multicolumn{1}{l}{United States} & 10    & 10    & 10 \\
    Ireland & 10    & 10    & (5, 6) & \multicolumn{1}{l}{Australia} & (1, 2) & (1, 2) & (2, 3) \\
    Spain & (1, 2) & (1, 2) & (2, 3) & \multicolumn{1}{l}{New Zealand} & (1, 2) & (1, 2) & (6, 7) \\
    Portugal & (5, 6) & (4, 5) & (2, 3) & \multicolumn{1}{l}{South Korea} & (4, 5) & (5, 6) & (6, 7) \\
          &       &       &       & \multicolumn{1}{l}{Chile} & (6, 7) & (6, 7) & (6, 7) \\
    \multicolumn{4}{l}{\textbf{EFTA}} & \multicolumn{1}{l}{Japan} & 1     & 1     & (5, 6) \\
    \multicolumn{1}{p{5.91em}}{Norway} & (9, 10) & (9, 10) & (9, 10) & \multicolumn{1}{l}{Mexico} & (1, 2) & (4, 5) & 1 \\
    Iceland & 10    & 10    & 10    & \multicolumn{1}{l}{Turkey} & (7, 8) & (8, 9) & (6, 7) \\
    Switzerland & 10    & 10    & 10    &       &       &       &  \\
    \bottomrule
    \end{tabular}%
        \begin{tablenotes}
        \setlength\labelsep{0pt}
        \footnotesize
        \item Note: 1) Legend: 10 = above 95\textsuperscript{th} quantile; (9,10) = between 85\textsuperscript{th} and 95\textsuperscript{th} quantile; 
                    (8,9) = between 75\textsuperscript{th} and 85\textsuperscript{th} quantile; (7,8) = between 65\textsuperscript{th} and 75\textsuperscript{th} quantile;
                    (6,7) = between 55\textsuperscript{th} and 65\textsuperscript{th} quantile; (5,6) = between 45\textsuperscript{th} and 55\textsuperscript{th} quantile; (4,5) = between 35\textsuperscript{th} and 45\textsuperscript{th} quantile; (3,4) = between 25\textsuperscript{th} and 35\textsuperscript{th} quantile; (2,3) = between 15\textsuperscript{th} and 25\textsuperscript{th} quantile; (1,2) = between 5\textsuperscript{th} and 15\textsuperscript{th} quantile; 1 = below 5\textsuperscript{th} quantile. 2) Although the United Kingdom left the EU in January 2020, we still treat it as an EU country owing to the sample period.
        \end{tablenotes}
    \end{threeparttable}
\end{table}%

Several interesting findings can be drawn from Table \ref{tab:perf}. First, in general, EFTA countries performed best in sustainable development, while the ranking of Non-European countries was relatively lower than that of the other groups. With respect to time, there was an increase in sustainability (i.e., the level of performance) for most countries and a decrease for a few other countries. Although Table \ref{tab:perf} provides a similar ranking to the overall performance from the SDG Dashboard Reports (\citename{Sachs2017}, \citeyear*{Sachs2017}, \citeyear*{Sachs2019}, \citeyear*{Sachs2020}) and \citeasnoun{Lamichhane2020}, there was a minor difference in ranking for certain countries (c.f., the ranking of Germany) due to the fact that all 17 SDGs with 231 unique indicators were evaluated in their studies. Instead, we consider a large portion of environmental indicators as input variables (see Table \ref{tab:ind}).

Second, countries from the EU transition economies can serve as the benchmark. As shown in Table \ref{tab:perf}, the developed countries did not necessarily perform best in sustainable development, whereas the less-developed countries had a chance to become better (e.g., Estonia and Slovenia). This partly demonstrates the advantages of the CQR approach over the full frontier estimation in structure analysis. If a unit is located in the interior of a production probability set, it will be treated as inefficient in the traditional frontier, in contrast to quantile-based efficiency analysis, where we treat it as an efficient unit. Therefore, less-developed countries, such as the Czech Republic and Slovenia, could perform best in sustainable development. As stated earlier, this might be related to the selection of input-output variables.

Third, Table \ref{tab:perf} reveals the inequality of the relative performances of sustainability among the OECD countries. The large variation in expectiles would suggest that the plan, action, and implementation of facilitating sustainable development in OECD counties are inefficient. The inequality of sustainability has become a barrier to the 2030 Agenda for Sustainable Development for OECD countries, even to all countries worldwide. Consequently, for the upper-expectile countries, they can achieve the 17 SDGs easily, but it is difficult for lower-expectile countries to achieve those goals. Table \ref{tab:perf} also suggests that there is large room for OECD countries to improve in terms of sustainability. This considerable room for improvement calls for ambitious plans and actions to achieve the 17 SDGs. A whole plan, for example, in terms of SO$_2$ emissions, similar to the Paris Agreement, could be considered and implemented in OECD countries. 


\section{Conclusions}\label{sec:conc}

Addressing the curse of dimensionality under a limited sample size and sparse data space remains a challenge in production efficiency analysis. In this paper, we developed a new $\mathcal{L}_0$-norm regularization approach to CQR and CER for subset variable selection. The paper investigated the finite sample performance of the proposed $\mathcal{L}_0$-norm regularization in contrast to the commonly used $\mathcal{L}_1$-norm regularization via an MC study (Section \ref{sec:mc}). The proposed $\mathcal{L}_0$-norm regularization approach was further applied to the SDG data of OECD countries for the years 2017, 2019 and 2020 to evaluate the performance of sustainability and compare the methods' performances empirically (Section \ref{sec:sdg}). The proposed $\mathcal{L}_0$-norm regularization approach can more effectively address the curse of dimensionality via subset variable selection in multidimensional spaces.

The evidence from the MC study suggests that $\mathcal{L}_0$-CQR has a major advantage over $\mathcal{L}_1$-CQR in subset variable selection, especially in the high-SNR regime. $\mathcal{L}_1$-CQR is unable to completely eliminate the irrelevant variables compared with the true support set in certain cases, whereas $\mathcal{L}_0$-CQR can successfully select the true relevant variables at least partially. There is no systematic increasing or decreasing trend as the quantile varies. The effects of the SNR and sample size are also observed: the higher the SNR level is, the lower the prediction error and accuracy; the larger the sample size is, the lower the prediction error and accuracy. Furthermore, we find that the expectile approach generally performs better than the quantile approach. 

The SDG application results demonstrate that the $\mathcal{L}_0$-CER approach can eliminate more irrelevant variables than the $\mathcal{L}_1$-CER approach, indicating a better performance in subset selection. The relative expectile ranking also helps reveal the heterogeneity of performance in sustainable development. The results also suggest that the countries in the EFTA group outperform those in other groups, that the countries in the EU transition economies can serve as the benchmark, and that there is large inequality of the relative performance of sustainability among the OECD countries. 

The findings drawn from this study can provide insight into the relationship between $\mathcal{L}_1$-norm regularization and $\mathcal{L}_0$-norm regularization in the context of production economics, especially in the era of big data. Furthermore, the proposed approach is applicable to other relative performance evaluations (e.g., schools, hospitals, and electricity distribution system operators). Finding a more efficient procedure to determine the tuning parameter is one of the fascinating avenues for future research. The other, as noted in Section \ref{sec:l1cqr}, is to provide a formal proof that application of the $\mathcal{L}_1$-norm squared regularization to the CQR problem can ensure the uniqueness of the subgradients.


\section*{Acknowledgments}\label{sec:ack}
The author is indebted to Timo Kuosmanen for his valuable guidance, advice, and comments. I also gratefully acknowledge the computational resources provided by the Aalto Science-IT project and financial support from the Foundation for Economic Education (Liikesivistysrahasto) (nos. 180019, 190073) and the HSE Support Foundation (no. 11--2290).

\baselineskip 12pt
\bibliography{references} 

@article{Mazumder2019,
    title = {{A computational framework for multivariate convex regression and its variants}},
    year = {2019},
    journal = {Journal of the American Statistical Association},
    author = {Mazumder, Rahul and Choudhury, Arkopal and Iyengar, Garud and Sen, Bodhisattva},
    month = {1},
    pages = {318--331},
    volume = {114},
    publisher = {American Statistical Association},
    issn = {1537274X},
    arxivId = {1509.08165}
}

@article{Huan2021,
    title = {{A method for assessing the impacts of an international agreement on regional progress towards Sustainable Development Goals}},
    year = {2021},
    journal = {Science of the Total Environment},
    author = {Huan, Yizhong and Yu, Yurong and Liang, Tao and Burgman, Mark},
    pages = {147336},
    volume = {785}
}

@article{Lee2013,
    title = {{A more efficient algorithm for Convex Nonparametric Least Squares}},
    year = {2013},
    journal = {European Journal of Operational Research},
    author = {Lee, Chia Yen and Johnson, Andrew L. and Moreno-Centeno, Erick and Kuosmanen, Timo},
    pages = {391--400},
    volume = {227},
    isbn = {0377-2217},
    issn = {03772217}
}

@article{Sinha2015,
    title = {{A multiobjective exploratory procedure for regression model selection}},
    year = {2015},
    journal = {Journal of Computational and Graphical Statistics},
    author = {Sinha, Ankur and Malo, Pekka and Kuosmanen, Timo},
    number = {1},
    month = {1},
    pages = {154--182},
    volume = {24},
    publisher = {American Statistical Association},
    issn = {15372715}
}

@article{Pastor2002,
    title = {{A statistical test for nested radial DEA models}},
    year = {2002},
    journal = {Operations Research},
    author = {Pastor, Jesús T. and Ruiz, José L. and Sirvent, Inmaculada},
    pages = {728--735},
    volume = {50},
    publisher = {INFORMS Inst.for Operations Res.and the Management Sciences},
    issn = {0030364X}
}

@article{Dula2011,
    title = {{An algorithm for data envelopment analysis}},
    year = {2011},
    journal = {INFORMS Journal on Computing},
    author = {Dul{\'{a}}, J. H.},
    month = {3},
    pages = {284--296},
    volume = {23},
    issn = {10919856}
}

@unpublished{Lin2020,
    title = {{An augmented Lagrangian method with constraint generations for shape-constrained convex regression problems}},
    year = {2020},
    author = {Lin, Meixia and Sun, Defeng and Toh, Kim-Chuan},
    month = {12},
    url = {http://arxiv.org/abs/2012.04862},
    arxivId = {2012.04862}
}

@article{Lamichhane2020,
    title = {{Benchmarking OECD countries’ sustainable development performance: A goal-specific principal component analysis approach}},
    year = {2020},
    journal = {Journal of Cleaner Production},
    author = {Lamichhane, Shyam and E{\u{g}}ilmez, Gökhan and Gedik, Ridvan and Bhutta, M. Khurrum S. and Erenay, Bulent},
    month = {3},
    pages = {125040},
    volume = {287},
    publisher = {Elsevier Ltd},
    issn = {09596526}
}

@article{Bertsimas2016,
    title = {{Best subset selection via a modern optimization lens}},
    year = {2016},
    journal = {Annals of Statistics},
    author = {Bertsimas, Dimitris and King, Angela and Mazumder, Rahul},
    month = {4},
    pages = {813--852},
    volume = {44},
    publisher = {Institute of Mathematical Statistics},
    issn = {00905364}
}

@article{Hastie2020,
    title = {{Best subset, forward stepwise or Lasso? Analysis and recommendations based on extensive comparisons}},
    year = {2020},
    journal = {Statistical Science},
    author = {Hastie, Trevor and Tibshirani, Robert and Tibshirani, Ryan},
    month = {11},
    pages = {579--592},
    volume = {35},
    publisher = {Institute of Mathematical Statistics},
    issn = {0883-4237}
}

@article{Chen2021a,
    title = {{Binary classification with covariate selection through L0-penalised empirical risk minimisation}},
    year = {2021},
    journal = {Econometrics Journal},
    author = {Chen, Le-Yu and Lee, Sokbae},
    pages = {103--120},
    volume = {24}
}

@article{Lavergne2008,
    title = {{Breaking the curse of dimensionality in nonparametric testing}},
    year = {2008},
    journal = {Journal of Econometrics},
    author = {Lavergne, Pascal and Patilea, Valentin},
    month = {3},
    pages = {103--122},
    volume = {143},
    issn = {03044076}
}

@article{Wilson2018,
    title = {{Dimension reduction in nonparametric models of production}},
    year = {2018},
    journal = {European Journal of Operational Research},
    author = {Wilson, Paul W.},
    month = {5},
    pages = {349--367},
    volume = {267},
    publisher = {Elsevier B.V.},
    issn = {03772217}
}

@article{Lai2018,
    title = {{Estimation of the production profile and metafrontier technology gap: A quantile approach}},
    year = {2018},
    journal = {Empirical Economics},
    author = {Lai, Hung pin and Huang, Cliff J. and Fu, Tsu Tan},
    month = {11},
    pages = {2709--2731},
    volume = {58},
    publisher = {Springer Verlag},
    issn = {03777332}
}

@article{Su2017,
    title = {{False discoveries occur early on the lasso path}},
    year = {2017},
    journal = {Annals of Statistics},
    author = {Su, Weijie and Bogdan, Malgorzata and Cand{\`{e}}s, Emmanuel},
    pages = {2133--2150},
    volume = {45},
    isbn = {202112:33:08},
    issn = {00905364}
}

@article{Benitez-Pena2019,
    title = {{Feature selection in data envelopment analysis: A mathematical optimization approach}},
    year = {2020},
    journal = {Omega},
    author = {Ben{\'{i}}tez-Pe{\~{n}}a, Sandra and Bogetoft, Peter and Romero Morales, Dolores},
    pages = {102068},
    volume = {96},
    publisher = {Elsevier Ltd},
    issn = {03050483}
}

@article{Dai2020,
    title = {{Forward-looking assessment of the GHG abatement cost: Application to China}},
    year = {2020},
    journal = {Energy Economics},
    author = {Dai, Sheng and Zhou, Xun and Kuosmanen, Timo},
    month = {5},
    pages = {104758},
    volume = {88},
    issn = {01409883}
}

@article{Nataraja2011,
    title = {{Guidelines for using variable selection techniques in data envelopment analysis}},
    year = {2011},
    journal = {European Journal of Operational Research},
    author = {Nataraja, Niranjan R. and Johnson, Andrew L.},
    month = {12},
    pages = {662--669},
    volume = {215},
    publisher = {North-Holland},
    issn = {03772217}
}

@article{Kuosmanen2020,
    title = {{How much climate policy has cost for OECD countries?}},
    year = {2020},
    journal = {World Development},
    author = {Kuosmanen, Timo and Zhou, Xun and Dai, Sheng},
    month = {1},
    pages = {104681},
    volume = {125},
    issn = {0305750X}
}

@article{Adler2010a,
    title = {{Improving discrimination in data envelopment analysis: PCA-DEA or variable reduction}},
    year = {2010},
    journal = {European Journal of Operational Research},
    author = {Adler, Nicole and Yazhemsky, Ekaterina},
    month = {4},
    pages = {273--284},
    volume = {202},
    publisher = {North-Holland},
    issn = {03772217}
}

@unpublished{Qin2014,
    title = {{Joint variable selection for data envelopement analysis via group sparsity}},
    year = {2014},
    author = {Qin, Zhiwei and Song, Irene},
    month = {2},
    url = {http://arxiv.org/abs/1402.3740},
    arxivId = {1402.3740}
}

@article{Dunning2017,
    title = {{JuMP: A modeling language for mathematical optimization}},
    year = {2017},
    journal = {SIAM review},
    author = {Dunning, Iain and Huchette, Joey and Lubin, Miles},
    pages = {295--320},
    volume = {59},
    publisher = {SIAM},
    issn = {0036-1445}
}

@article{Lee2020,
    title = {{LASSO variable selection in data envelopment analysis with small datasets}},
    year = {2020},
    journal = {Omega},
    author = {Lee, Chia Yen and Cai, Jia Ying},
    month = {3},
    pages = {102019},
    volume = {91},
    publisher = {Elsevier Ltd},
    issn = {03050483}
}

@article{Chen2021,
    title = {{LASSO+DEA for small and big wide data}},
    year = {2021},
    journal = {Omega},
    author = {Chen, Ya and Tsionas, Mike G. and Zelenyuk, Valentin},
    month = {1},
    pages = {102419},
    volume = {102},
    publisher = {Pergamon},
    issn = {03050483}
}

@article{Dalampira2020,
    title = {{Mapping sustainable development goals: A network analysis framework}},
    year = {2020},
    journal = {Sustainable Development},
    author = {Dalampira, Evropi Sofia and Nastis, Stefanos A.},
    month = {1},
    pages = {46--55},
    volume = {28},
    publisher = {John Wiley and Sons Ltd},
    issn = {10991719}
}

@article{Costanza2016,
    title = {{Modelling and measuring sustainable wellbeing in connection with the UN Sustainable Development Goals}},
    year = {2016},
    journal = {Ecological Economics},
    author = {Costanza, Robert and Daly, Lew and Fioramonti, Lorenzo and Giovannini, Enrico and Kubiszewski, Ida and Mortensen, Lars Fogh and Pickett, Kate E. and Ragnarsdottir, Kristin Vala and De Vogli, Roberto and Wilkinson, Richard},
    month = {10},
    pages = {350--355},
    volume = {130},
    publisher = {Elsevier B.V.},
    issn = {09218009}
}

@inproceedings{Balazs2015,
    title = {{Near-optimal max-affine estimators for convex regression}},
    year = {2015},
    booktitle = {18th Artificial Intelligence and Statistics},
    author = {Bal{\'{a}}zs, Gábor and Gy{\"{o}}rgy, András and Szepesv{\'{a}}ri, Csaba},
    pages = {38:56–64},
    publisher = {PMLR}
}

@article{Wang2014c,
    title = {{Nonparametric quantile frontier estimation under shape restriction}},
    year = {2014},
    journal = {European Journal of Operational Research},
    author = {Wang, Yongqiao and Wang, Shouyang and Dang, Chuangyin and Ge, Wenxiu},
    pages = {671--678},
    volume = {232},
    publisher = {Elsevier B.V.},
    isbn = {03772217},
    issn = {03772217}
}

@article{Chen2020c,
    title = {{On degrees of freedom of projection estimators with applications to multivariate nonparametric regression}},
    year = {2020},
    journal = {Journal of the American Statistical Association},
    author = {Chen, Xi and Lin, Qihang and Sen, Bodhisattva},
    number = {529},
    month = {1},
    pages = {173--186},
    volume = {115},
    publisher = {American Statistical Association},
    issn = {1537274X},
    arxivId = {1509.01877}
}

@article{Stone2007,
    title = {{Optimal rates of convergence for nonparametric estimators}},
    year = {1980},
    journal = {The Annals of Statistics},
    author = {Stone, Charles J},
    pages = {1348--1360},
    volume = {8},
    isbn = {202110:45:32},
    issn = {0090-5364}
}

@article{Dyson2001,
    title = {{Pitfalls and protocols in DEA}},
    year = {2001},
    journal = {European Journal of Operational Research},
    author = {Dyson, R. G. and Allen, R. and Camanho, A. S. and Podinovski, V. V. and Sarrico, C. S. and Shale, E. A.},
    month = {7},
    pages = {245--259},
    volume = {132},
    publisher = {North-Holland},
    issn = {03772217}
}

@article{Behr2010,
    title = {{Quantile regression for robust bank efficiency score estimation}},
    year = {2010},
    journal = {European Journal of Operational Research},
    author = {Behr, Andreas},
    pages = {568--581},
    volume = {200},
    publisher = {Elsevier B.V.},
    issn = {03772217}
}

@article{Efron1991,
    title = {{Regression percentiles using asymmetric squared error loss}},
    year = {1991},
    journal = {Statistica Sinica},
    author = {Efron, B},
    pages = {93--125},
    volume = {1}
}

@article{Koenker1978,
    title = {{Regression quantiles}},
    year = {1978},
    journal = {Econometrica},
    author = {Koenker, Roger and Bassett, Gilbert},
    pages = {33--50},
    volume = {46},
    publisher = {JSTOR},
    issn = {0012-9682}
}

@article{Tibshirani1996,
    title = {{Regression shrinkage and selection via the Lasso}},
    year = {1996},
    journal = {Journal of the Royal Statistical Society. Series B (Methodological)},
    author = {Tibshirani, Robert},
    pages = {267--288},
    volume = {58}
}

@article{Kuosmanen2008,
    title = {{Representation theorem for convex nonparametric least squares}},
    year = {2008},
    journal = {Econometrics Journal},
    author = {Kuosmanen, Timo},
    month = {7},
    pages = {308--325},
    volume = {11},
    issn = {13684221}
}

@book{Sachs2017,
    title = {{SDG index and dashboards report 2017}},
    year = {2017},
    booktitle = {New York: Bertelsmann Stiftung and Sustainable Development Solutions Network (SDSN)},
    author = {Sachs, Jeffrey and Schmidt-Traub, Guido and Kroll, Christian and Durand-Delacre, David and Teksoz, Katerina},
    publisher = {Bertelsmann Stiftung and Sustainable Development Solutions Network (SDSN)},
    address = {New York}
}

@article{Keshvari2018,
    title = {{Segmented concave least squares: A nonparametric piecewise linear regression}},
    year = {2018},
    journal = {European Journal of Operational Research},
    author = {Keshvari, Abolfazl},
    month = {4},
    pages = {585--594},
    volume = {266},
    publisher = {Elsevier B.V.},
    issn = {03772217}
}

@article{Kuosmanen2020b,
    title = {{Shadow prices and marginal abatement costs: Convex quantile regression approach}},
    year = {2021},
    journal = {European Journal of Operational Research},
    author = {Kuosmanen, Timo and Zhou, Xun},
    month = {7},
    pages = {666--675},
    volume = {289},
    publisher = {Elsevier BV},
    issn = {03772217}
}

@article{Natarajan1995,
    title = {{Sparse approximate solutions to linear systems}},
    year = {1995},
    journal = {SIAM Journal on Computing},
    author = {Natarajan, B K},
    pages = {227--234},
    volume = {24},
    issn = {00975397}
}

@unpublished{Bertsimas2017,
    title = {{Sparse Classification: A scalable discrete optimization perspective}},
    year = {2017},
    author = {Bertsimas, Dimitris and Pauphilet, Jean and Van Parys, Bart},
    month = {10},
    url = {https://arxiv.org/abs/1710.01352},
    arxivId = {1710.01352}
}

@article{Bertsimas2020,
    title = {{Sparse convex regression}},
    year = {2021},
    journal = {INFORMS Journal on Computing},
    author = {Bertsimas, Dimitris and Mundru, Nishanth},
    month = {6},
    pages = {262--279},
    volume = {33},
    publisher = {Institute for Operations Research and the Management Sciences (INFORMS)},
    issn = {15265528}
}

@article{Bertsimas2020b,
    title = {{Sparse high-dimensional regression: Exact scalable algorithms and phase transitions}},
    year = {2020},
    journal = {Annals of Statistics},
    author = {Bertsimas, Dimitris and van Parys, Bart},
    pages = {300--323},
    volume = {48},
    publisher = {Institute of Mathematical Statistics},
    issn = {21688966},
    arxivId = {1709.10029}
}

@article{Wagner2007,
    title = {{Stepwise selection of variables in data envelopment analysis: Procedures and managerial perspectives}},
    year = {2007},
    journal = {European Journal of Operational Research},
    author = {Wagner, Janet M. and Shimshak, Daniel G.},
    month = {7},
    pages = {57--67},
    volume = {180},
    publisher = {North-Holland},
    issn = {03772217}
}

@incollection{Kuosmanen2015d,
    title = {{Stochastic nonparametric approach to efficiency analysis: A unified framework}},
    year = {2015},
    booktitle = {Data Envelopment Analysis},
    author = {Kuosmanen, Timo and Johnson, Andrew and Saastamoinen, Antti},
    editor = {Zhu, Joe},
    chapter = {7},
    pages = {191--244},
    publisher = {Springer},
    address = {Boston, MA}
}

@article{Fan2010,
    title = {{Sure independence screening in generalized linear models with NP-dimensionality}},
    year = {2010},
    journal = {Annals of Statistics},
    author = {Fan, Jianqing and Song, Rui},
    month = {12},
    pages = {3567--3604},
    volume = {38},
    issn = {00905364}
}

@book{Sachs2019,
    title = {{Sustainable development report 2019}},
    year = {2019},
    booktitle = {New York: Bertelsmann Stiftung and Sustainable Development Solutions Network (SDSN)},
    author = {Sachs, Jeffrey and Schmidt-Traub, Guido and Kroll, Christian and Lafortune, Guillaume and Fuller, Grayson},
    publisher = {Bertelsmann Stiftung and Sustainable Development Solutions Network (SDSN)},
    address = {New York}
}

@book{Sachs2020,
    title = {{The Sustainable Development Goals and COVID-19. Sustainable Development Report 2020}},
    year = {2020},
    booktitle = {Sustainable Development Report 2020},
    author = {Sachs, Jeffrey and Schmidt-Traub, Guido and Kroll, Christian and Lafortune, Guillaume and Fuller, Grayson and Woelm, Finn},
    publisher = {Cambridge University Press},
    address = {Cambridge}
}

@article{Singpai,
    title = {{Using a DEA-AutoML approach to track SDG achievements}},
    year = {2020},
    journal = {Sustainability},
    author = {Singpai, Bodin and Wu, Desheng},
    pages = {1--26},
    volume = {12},
    issn = {20711050}
}

@article{Homburg2001,
    title = {{Using data envelopment analysis to benchmark activities}},
    year = {2001},
    journal = {International Journal of Production Economics},
    author = {Homburg, Carsten},
    month = {8},
    pages = {51--58},
    volume = {73},
    publisher = {Elsevier},
    issn = {09255273}
}

@article{Li2017,
    title = {{Variable selection in data envelopment analysis via Akaike's information criteria}},
    year = {2017},
    journal = {Ann Oper Res},
    author = {Li, Yongjun and Shi, Xiao and Yang, Min and Liang, Liang},
    pages = {453--476},
    volume = {253}
}

\clearpage
\newpage
\baselineskip 20pt
\section*{Appendix}\label{sec:app}


\renewcommand{\thesubsection}{\Alph{subsection}}
\renewcommand{\thefigure}{A\arabic{figure}}
\setcounter{figure}{0}
\renewcommand{\thetable}{A\arabic{table}}
\setcounter{table}{0}
\renewcommand{\theequation} {A.\arabic{equation}}
\setcounter{equation}{0}
\setcounter{footnote}{0}

\subsection{The CNLS-A algorithm}\label{sec:anlsa}

The CNLS-A algorithm, an extension of \citeasnoun{Bertsimas2020}, is designed for improving the computational efficiency in solving the penalized CQR and CER problems. Specifically, to speed up the calculation, the following strategy is applied: solving the reduced master problem 
that contains only a few constraints and then iteratively adding the violated constraints in a delayed manner.\footnote{As in \citeasnoun{Bertsimas2020}, we use the same notations: the master problem is the original problem with $n^2$ linear constraints (i.e., the formulation (\ref{eq:a1})); the reduced master problem is a problem with the same objective function and decision variables but with a subset of constraints.
}

For the sake of elaboration, the CQR problem is taken as an example to compare the computational efficiency between the CNLS-A and CNLS-G algorithms. Therefore, the CQR problem is defined as (\citename{Kuosmanen2015d}, \citeyear*{Kuosmanen2015d}) 
\begin{alignat}{2}
 \underset{\mathbf{\alpha},\mathbf{\bbeta },{{\mathbf{\varepsilon }}^{\text{+}}},{{\mathbf{\varepsilon }}^{-}}}{\mathop{\min }}&\,\tau \sum\limits_{i=1}^{n}{\varepsilon _{i}^{+}}+(1-\tau )\sum\limits_{i=1}^{n}{\varepsilon _{i}^{-}}  &{}& \label{eq:a1}\\
\mbox{\textit{s.t.}}\quad
& y_i=\mathbf{\alpha}_i+ \bbeta_i^{'}\bx_i+\varepsilon^{+}_i - \varepsilon^{-}_i &\quad& \forall i \notag\\
& \mathbf{\alpha}_i+\bbeta_{i}^{'}{{\bx}_{i}}\le \mathbf{\alpha}_h+\bbeta _h^{'}\bx_i  &{}& \forall i,h  \notag\\
& \bbeta_i\ge \bzero &{}& \forall i  \notag\\
& \varepsilon _i^{+}\ge 0,\ \varepsilon_i^{-} \ge 0 &{}& \forall i \notag
\end{alignat}

To adapt the existing cutting-plan algorithm (\citename{Bertsimas2020}, \citeyear*{Bertsimas2020}), the CQR problem first needs to be reformulated. The alternative CQR problem is
\begin{alignat}{2}
  \underset{\hat{y},\mathbf{\beta },{{\mathbf{\varepsilon }}^{\text{+}}},{{\mathbf{\varepsilon }}^{-}}}{\mathop{\min }}&\,\tau \sum\limits_{i=1}^{n}{\varepsilon _{i}^{+}}+(1-\tau )\sum\limits_{i=1}^{n}{\varepsilon _{i}^{-}} &{}& \label{eq:a2}\\ 
\mbox{\textit{s.t.}}\quad
& y_i - \hat{y}_i = \varepsilon _{i}^{+} - \varepsilon _{i}^{-} &\quad& \forall i \notag \\
& \hat{y}_i + \bbeta_i^{'}(\bx_j -\bx_i) \ge \hat{y}_j &\quad& \forall i, j \notag\\
& \bbeta_i\ge \bzero &{}& \forall i  \notag\\
& \varepsilon _i^{+}\ge 0,\ \varepsilon_i^{-} \ge 0 &{}& \forall i \notag
\end{alignat}
where $\hat{y}_i$, a new decision variable, is the estimated quantile function. Note that the alternative formulation \eqref{eq:a2} is derived from Problem \eqref{eq:a1} by applying a mathematical transformation. For the formulation \eqref{eq:a1}, $\hat{y}_i$ is used to replace $\mathbf{\alpha}_i+ \bbeta_i^{'}\bx_i$ in the first constraint, and the Afriat equality constraint is reconstructed by substituting $\hat{y}_i$ for the left-hand side and substituting out $\mathbf{\alpha}_j$ using $\mathbf{\alpha}_j = \hat{y}_j-\bbeta_j^{'}\bx_j$ for the right-hand side. 

There are two commonly used approaches to form the initial $n-1$ constraints: the spanning path (SP) based on the Euclidean distances between observations and the minimum spanning tree (MST) among observations. The initial reduced master problem of CQR is:
\begin{alignat}{2}
  \underset{\hat{y},\mathbf{\beta },{{\mathbf{\varepsilon }}^{\text{+}}},{{\mathbf{\varepsilon }}^{-}}}{\mathop{\min }}&\,\tau \sum\limits_{i=1}^{n}{\varepsilon _{i}^{+}}+(1-\tau )\sum\limits_{i=1}^{n}{\varepsilon _{i}^{-}} &{}&  \label{eq:a3}\\ 
\mbox{\textit{s.t.}}\quad
& \hat{y}_{i_1} + \bbeta_{i_1}^{'}(\bx_{i_2} -\bx_{i_1}) \ge \hat{y}_{i_2} \notag\\
& \hat{y}_{i_2} + \bbeta_{i_2}^{'}(\bx_{i_3} -\bx_{i_2}) \ge \hat{y}_{i_3} \notag \\
& \vdots \notag \\
& \hat{y}_{i_{n-1}} + \bbeta_{i_{n-1}}^{'}(\bx_{i_n} -\bx_{i_{n-1}}) \ge \hat{y}_{i_n} \notag\\
& y_i - \hat{y}_i = \varepsilon _{i}^{+} - \varepsilon _{i}^{-} &\quad& \forall i \notag \\
& \bbeta_i\ge \bzero &{}& \forall i \notag \\
& \varepsilon _i^{+}\ge 0,\ \varepsilon_i^{-} \ge 0 &{}& \forall i \notag
\end{alignat}

For a given solution (i.e., $\hat{y}, \bbeta^{'}_{i}, \varepsilon_i^{-}, \varepsilon_i^{+}$) to the reduced master problem \eqref{eq:a3}, we need to check whether it is a feasible solution to the master problem \eqref{eq:a2}. If yes, it is also the optimal solution for the master problem \eqref{eq:a2}; otherwise, we have to efficiently find a violated Afriat constraint. Finding a violated constraint is regarded as a separation problem (\citename{Bertsimas2020}, \citeyear*{Bertsimas2020}). Therefore, the separation problem finds the minimal index $m(i)$ for each observation $i$ and checks if the corresponding smallest value is less than zero. 
\begin{alignat}{2}
m(i) = \operatorname*{arg\,min}_{1 \le m \le n} \{\hat{y}_i - \hat{y}_m + \bbeta_{i}^{'}(\bx_m-\bx_i)\}&{}& \label{eq:a4}
\end{alignat}

Furthermore, the following violated constraint is added to the reduced master problem \eqref{eq:a3} only if it is smaller than a given tolerance.
\begin{alignat}{2}
\hat{y}_i + \bbeta_{i}^{'}(\bx_{m(i)}-\bx_i) \ge \hat{y}_{m(i)} &{}&  \label{eq:a5}
\end{alignat}

The $m$\textsuperscript{th} iteration by $T_m$ is then solved, and if $\min_{1 \le m \le n}\{\hat{y}_i - \hat{y}_m + \bbeta_{i}^{'}(\bx_m-\bx_i)\} \ge -tol$, a final optimal solution to the master problem \eqref{eq:a2} is found. The $m$\textsuperscript{th} iteration problem is as follows:
\begin{alignat}{2}
  \underset{\hat{y},\mathbf{\beta },{{\mathbf{\varepsilon }}^{\text{+}}},{{\mathbf{\varepsilon }}^{-}}}{\mathop{\min }}&\,\tau \sum\limits_{i=1}^{n}{\varepsilon _{i}^{+}}+(1-\tau )\sum\limits_{i=1}^{n}{\varepsilon _{i}^{-}} &{}&  \label{eq:a6}\\ 
\mbox{\textit{s.t.}}\quad
& \hat{y}_i + \bbeta_i^{'}(\bx_j -\bx_i) \ge \hat{y}_j &\quad& \forall (i, j) \in T_0 \notag \\
& \hat{y}_i + \bbeta_i^{'}(\bx_j -\bx_i) \ge \hat{y}_j &\quad& \forall (i, j) \in T_1 \notag\\
& \vdots \notag \\
& \hat{y}_i + \bbeta_i^{'}(\bx_j -\bx_i) \ge \hat{y}_j &\quad& \forall (i, j) \in T_m \notag
\end{alignat}

The CNLS-A algorithm procedure is summarized as:
\vspace{0.5em}

\begin{algorithm}[H]\label{alg1} 
  \KwData{$\{(\bx_i, y_i) \in \real^d \times \real: i =1,\cdots, n\}$ and $tol = 0.01$}
  $\mathrm{out} = 0$ and $m = 0$\;
  Solve the reduced master problem \eqref{eq:a6}\;
  \While{$\mathrm{out} = 0$}{
    \For{$1 \le i \le n$}{
    Solve the separation problem \eqref{eq:a4} to find a minimal index $m(i)$\;
    Add the corresponding violated constraint \eqref{eq:a5} to Problem \eqref{eq:a6}\;
    }
    \eIf{$\min_{1 \le m \le n}\{\hat{y}_i - \hat{y}_m + \bbeta_{i}^{'}(\bx_m-\bx_i)\} < -tol$}{
    Resolve the reduced master problem \eqref{eq:a6} with the newly added constraint set $T_{m+1}$\;
    }{
    $\mathrm{out} = 1$\;
    }
    $m=m+1$\;
 }
 \KwResult{$\hat{y}, \bbeta^{'}_{i}, \varepsilon_i^{-}, \varepsilon_i^{+}$ }
\caption{CNLS-A algorithm for solving the CQR problem \eqref{eq:a2}.}
\end{algorithm}

\bigskip

Fig.~\ref{fig:figA1} depicts a performance comparison between the CNLS-G and CNLS-A algorithms in solving the CQR and CER problems.\footnote{In this paper, we slightly adapt the CNLS-G algorithm that is design for the usual \texttt{CNLS} problem (\citename{Lee2013}, \citeyear*{Lee2013}) to solve the CQR problem, and we use our developed pyStoNED package (https://github.com/ds2010/pyStoNED) to run the CNLS-G algorithm.
}
Since there is no systematic difference in quantile (see the MC study in Section \ref{sec:mc}), we consider only the effect of the number of observations and the number of input variables. We repeat the simulation 10 times to average the metrics of interest (i.e., running time and number of constraints). Furthermore, we use the MST approach to form the initial reduced constraints. 

As expected, the CNLS-A algorithm requires less time to solve the CQR and CER problems than the CNLS-G algorithm does. In contrast to the running time of CNLS-A, CNLS-G increases rapidly as the input dimension ($d$) increases, especially when the high-dimensional data are estimated with a large sample size ($n>300$). For example, according to the upper-right panel in Fig.~\ref{fig:figA1}, we find that the CNLS-G algorithm takes 25 times and 124 times longer than the CNLS-A algorithm in solving the CQR problem with $n=100$ and $n=500$, respectively. Compared to the CQR approach, CER runs faster in both algorithms due to the stronger convexity in the objective function. 
\begin{figure}[!htbp]
    \centering
    \includegraphics[width=1\textwidth]{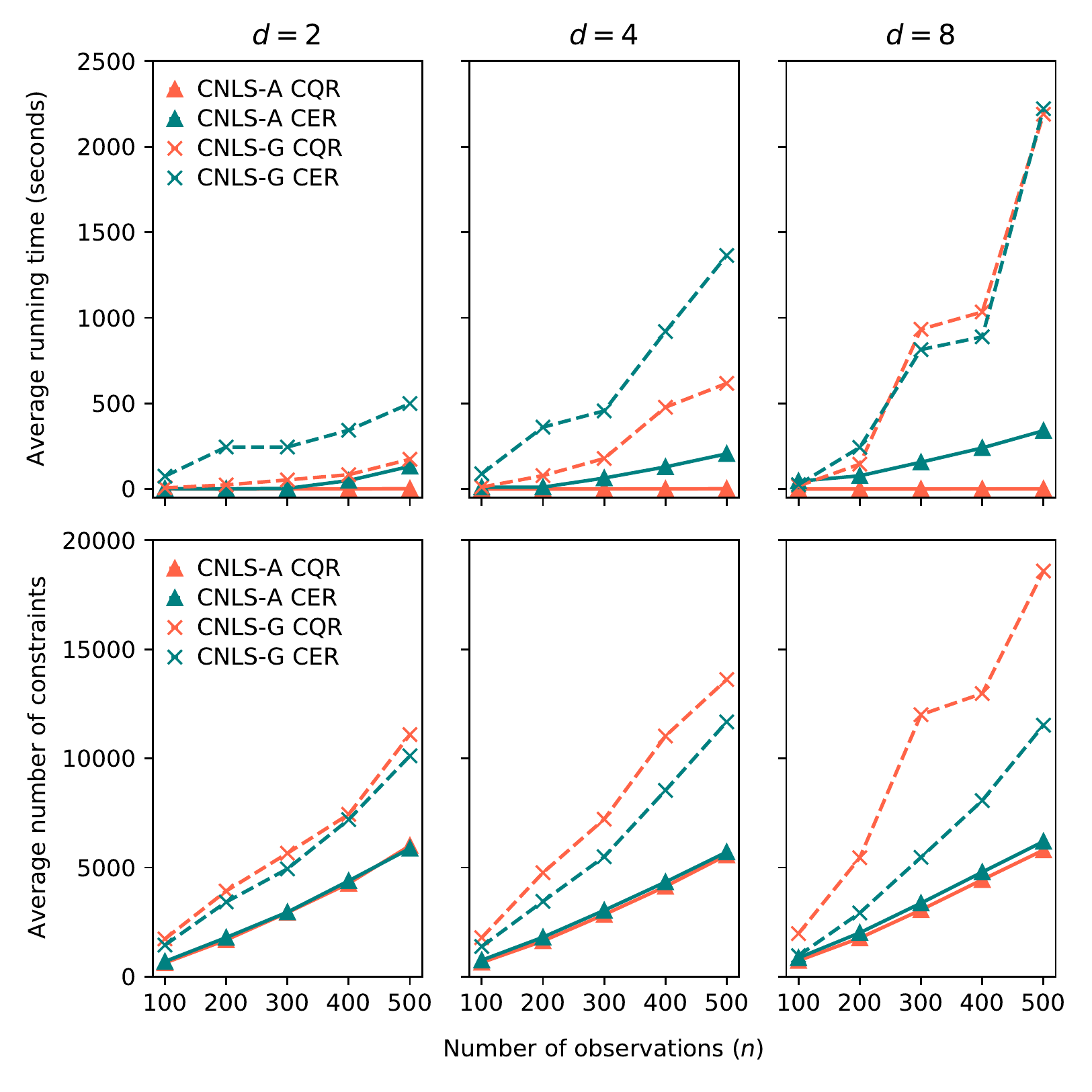}
    \caption{Performance comparison of the CNLS-A and CNLS-G algorithms to solve the CQR and CER problems with $\tau = 0.9$.}
    \label{fig:figA1}
\end{figure}

Regarding the average number of constraints, similar to the average running time, the CNLS-A algorithm can more efficiently find the violated constraints than can the CNLS-G algorithm in all the scenarios. CNLS-G includes an incremental increase in the constraints in the optimization as the number of input dimensions or the number of observations increases. By contrast, the constraints of the CNLS-A algorithm undergo a very small change. Furthermore, the CER approach uses less constraints than does the CQR approach. 


\newpage
\subsection{Additional tables and figures}
\renewcommand{\thefigure}{B\arabic{figure}}
\setcounter{figure}{0}
\renewcommand{\thetable}{B\arabic{table}}
\setcounter{table}{0}

\begin{table}[ht]
    \renewcommand\arraystretch{1.15}
    \caption{Summary of statistics for the input and output variables}\label{tab:stats}
    \centering
    \begin{threeparttable}
    \begin{tabular}{lllrrrrr}
    \toprule
    Variable   &    & Unit  & Mean  & Median & Min.   & Max.   & Std.Dev. \\
    \midrule
       &   \vspace{-1.5em}\\  
    {\bf Economic:}    \\
            $y$ (GDP)&       & 10$^{12}$ international \$     & 1.69  & 0.50  & 0.02  & 20.59  & 3.44  \\
                &   \vspace{-1.5em}\\   
    {\bf Social:}     \\    
            $x_1$    &    & \%        & 0.44  & 0.29  & 0.00  & 2.18  & 0.44  \\
            $x_2$    &    & years     & 77.53  & 79.20  & 66.20  & 84.20  & 5.16  \\
            $x_3$    &    & ladder score, 0-10 & 6.70  & 6.88  & 5.19  & 7.86  & 0.70  \\
            $x_4$ (Employment)    &    & 10$^9$ persons         & 16.38  & 4.63  & 0.14  & 153.55  & 27.40  \\
            $x_5$    &    & 1-100         & 35.21  & 33.74  & 25.59  & 57.83  & 6.94  \\
            &   \vspace{-1.5em}\\   
    {\bf Environmental:}    \\    
            $x_6$    &    & $\mu$g/m$^3$  & 13.79  & 12.03  & 5.40  & 44.31  & 7.27  \\
            $x_7$    &    & kg/day/capita & 1.91   & 1.89   & 0.88  & 4.54   & 0.65  \\
            $x_8$    &    & kg/capita     & 18.19  & 19.80  & 6.50  & 28.50  & 5.17  \\
            $x_9$    &    & kg/capita     & 42.22  & 25.15  & 1.70  & 344.94 & 50.90  \\
            $x_{10}$ &    & kg/capita     & 46.35  & 42.27  & 25.19 & 139.80 & 19.84  \\
            $x_{11}$ &    & tCO$_2$/capita& 7.96   & 7.36   & 3.44  & 18.70  & 3.65  \\
            $x_{12}$ &    & EUR/tCO$_2$   & 18.38  & 12.47  & 0.01  & 66.95  & 16.12  \\
    \bottomrule
    \end{tabular}%
    \begin{tablenotes}
        \setlength\labelsep{0pt}
        \footnotesize
        \item Notes: 1) To calculate the absolute values of the population and GDP, we also collect two additional indicators: GDP and PPP (constant 2017 international \$).
    \end{tablenotes}
    \end{threeparttable}
\end{table}

\begin{figure}[!htbp]
    \centering
    \includegraphics[width=0.95\textwidth]{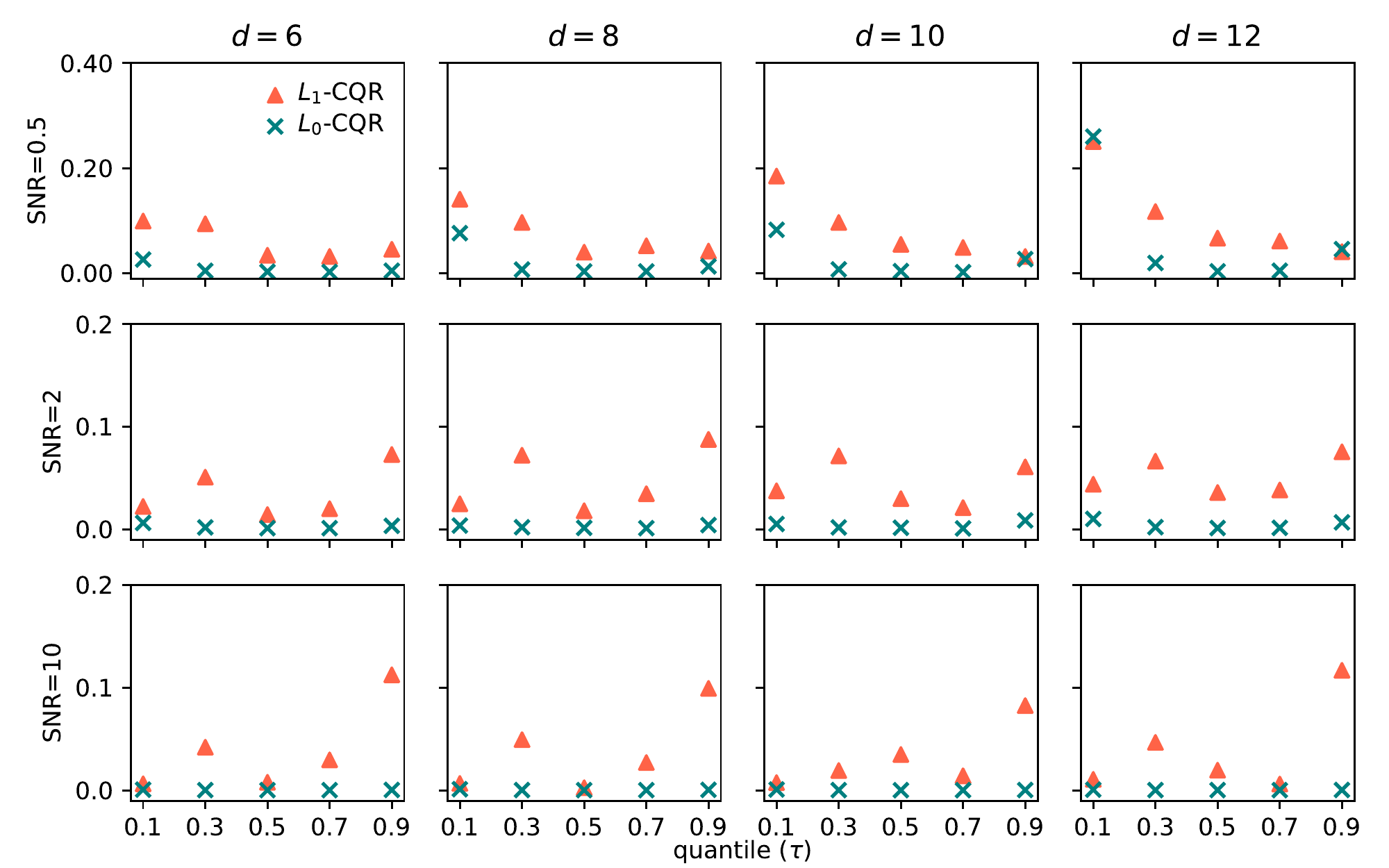}
    \caption{Prediction errors of the $\mathcal{L}_1$-CQR and $\mathcal{L}_0$-CQR approaches with $n = 500$ and $k_{\text{true}} = 2$.}
    \label{fig:figb1}
    \vspace{-1em}
\end{figure}

\begin{figure}[!htbp]
    \centering
    \includegraphics[width=0.95\textwidth]{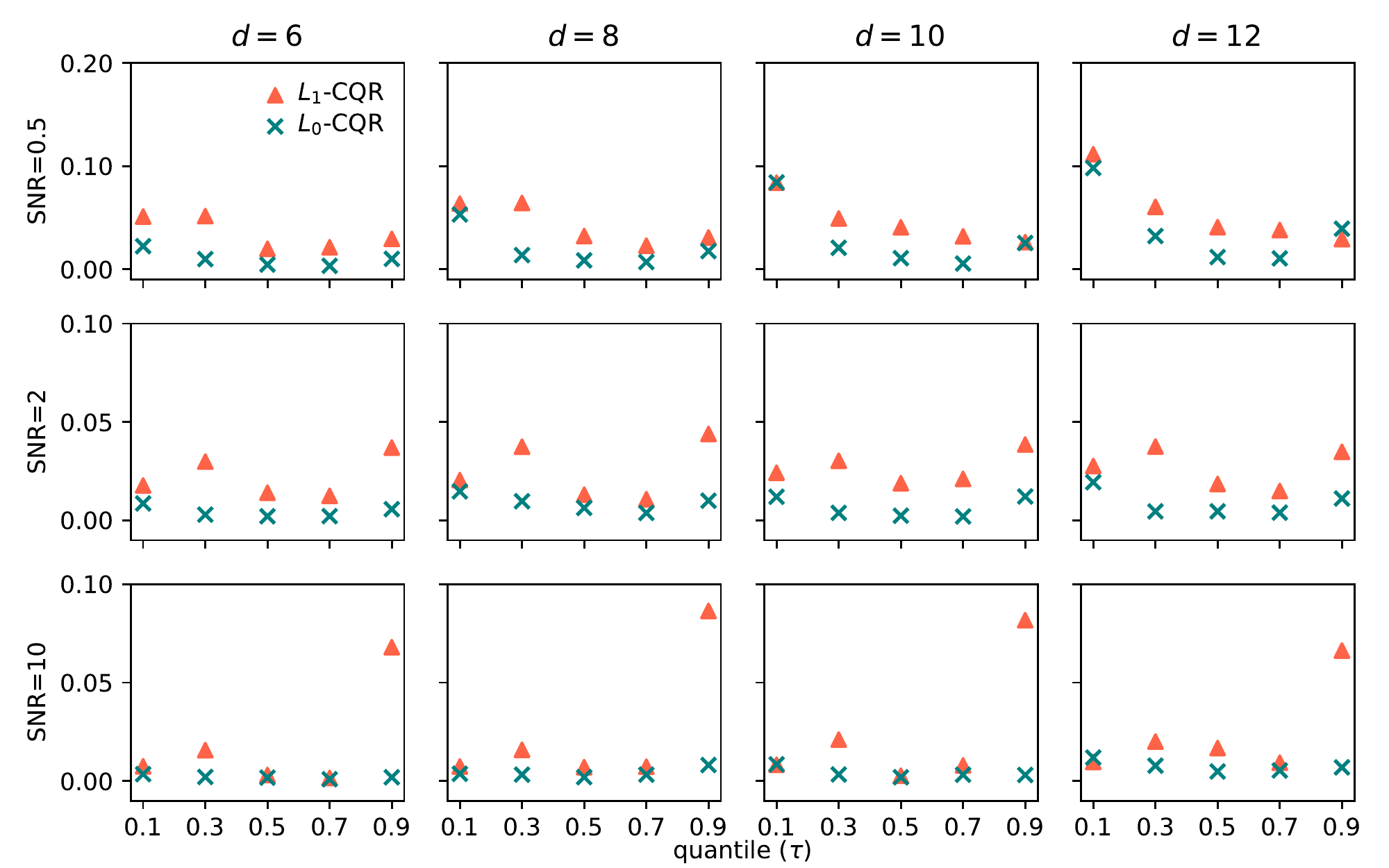}
    \caption{Prediction errors of the $\mathcal{L}_1$-CQR and $\mathcal{L}_0$-CQR approaches with $n = 500$ and $k_{\text{true}} = 4$.}
    \label{fig:figb2}
    \vspace{-1em}
\end{figure}

\begin{figure}[!htbp]
    \centering
    \includegraphics[width=0.95\textwidth]{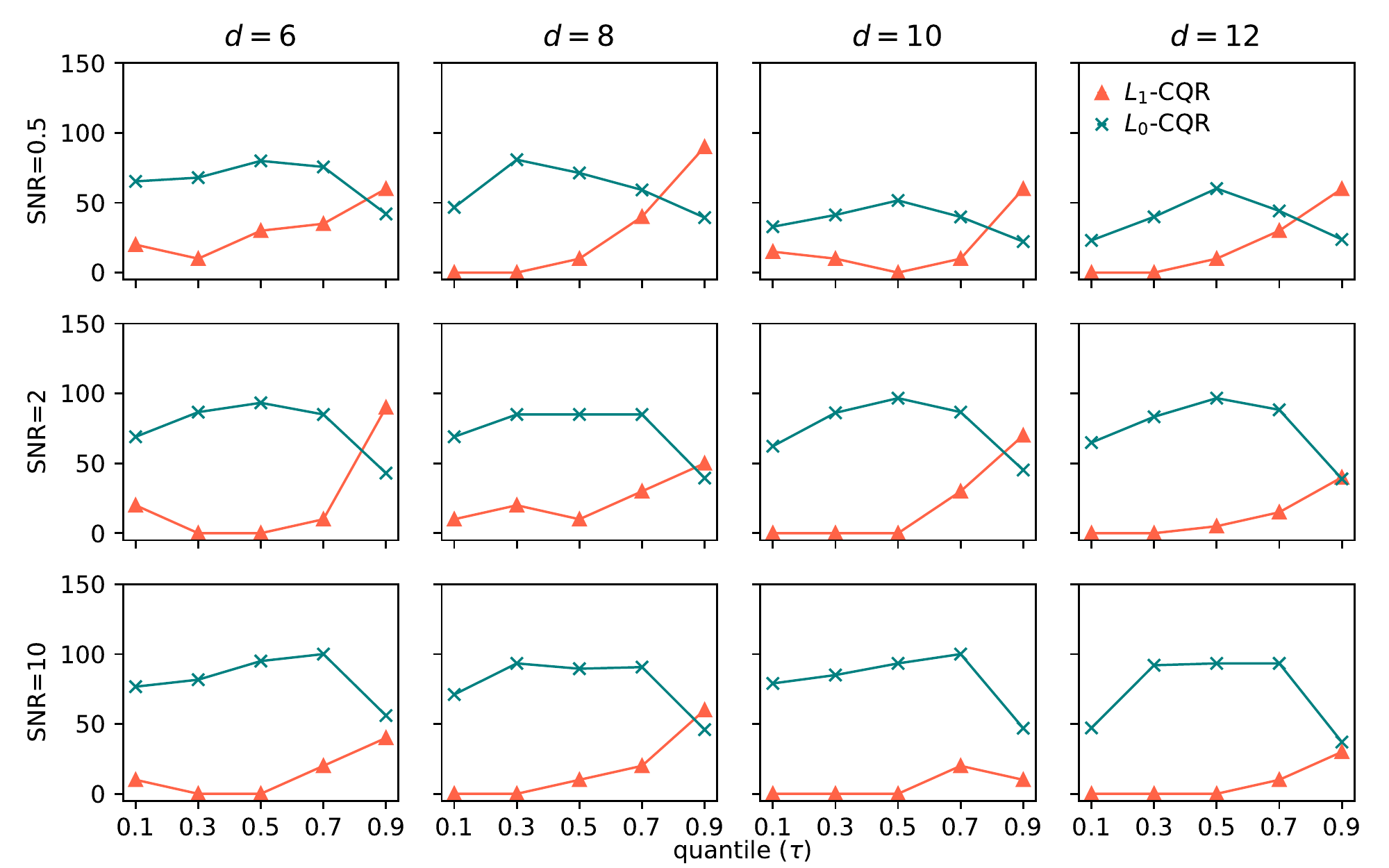}
    \caption{Accuracies of the $\mathcal{L}_1$-CQR and $\mathcal{L}_0$-CQR approaches with $n = 100$ and $k_{\text{true}} = 2$.}
    \label{fig:figb3}
    \vspace{-1em}
\end{figure}

\begin{figure}[!htbp]
    \centering
    \includegraphics[width=0.95\textwidth]{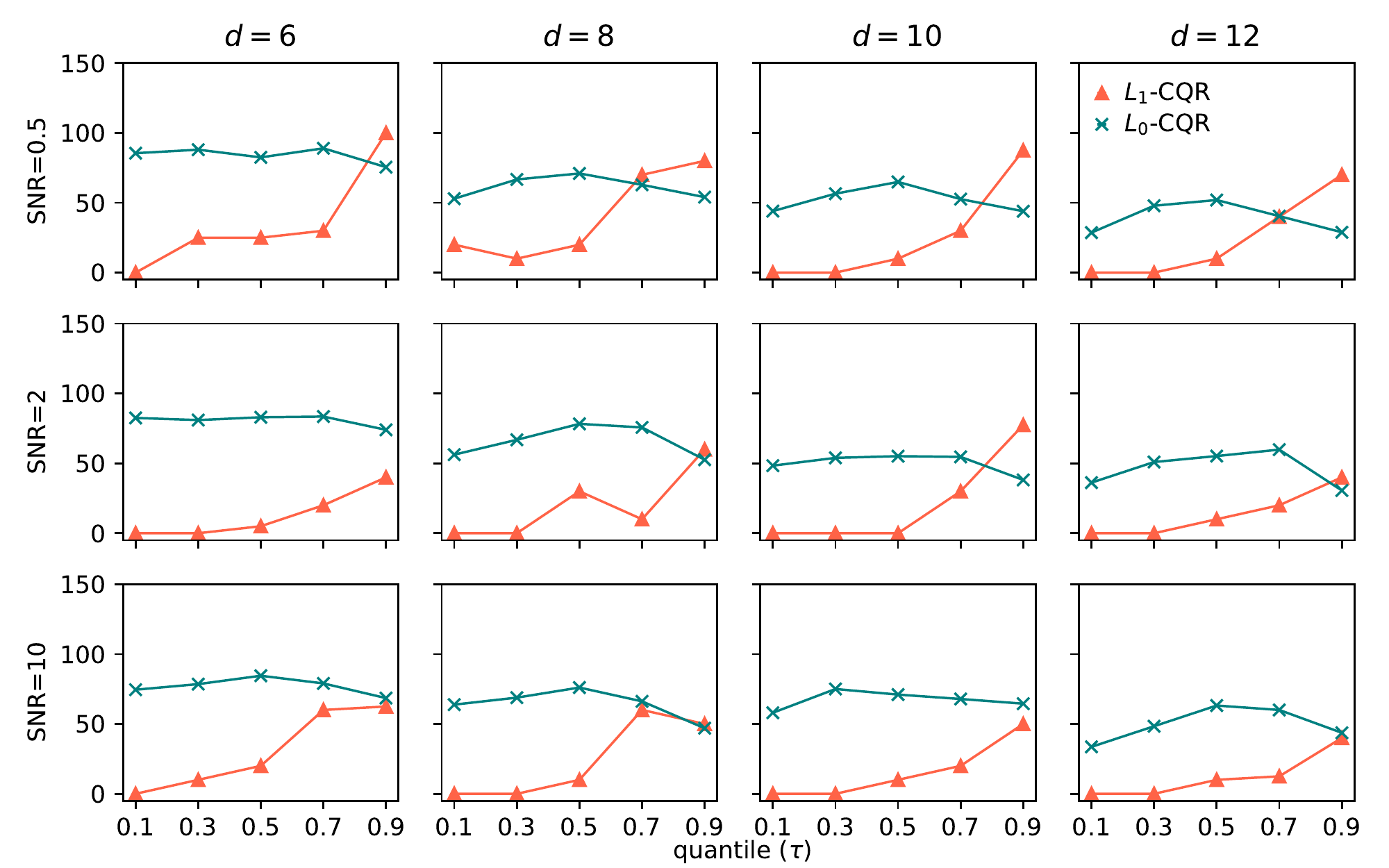}
    \caption{Accuracies of the $\mathcal{L}_1$-CQR and $\mathcal{L}_0$-CQR approaches with $n = 100$ and $k_{\text{true}} = 4$.}
    \label{fig:figb4}
    \vspace{-1em}
\end{figure}

\begin{figure}[!htbp]
    \centering
    \includegraphics[width=0.95\textwidth]{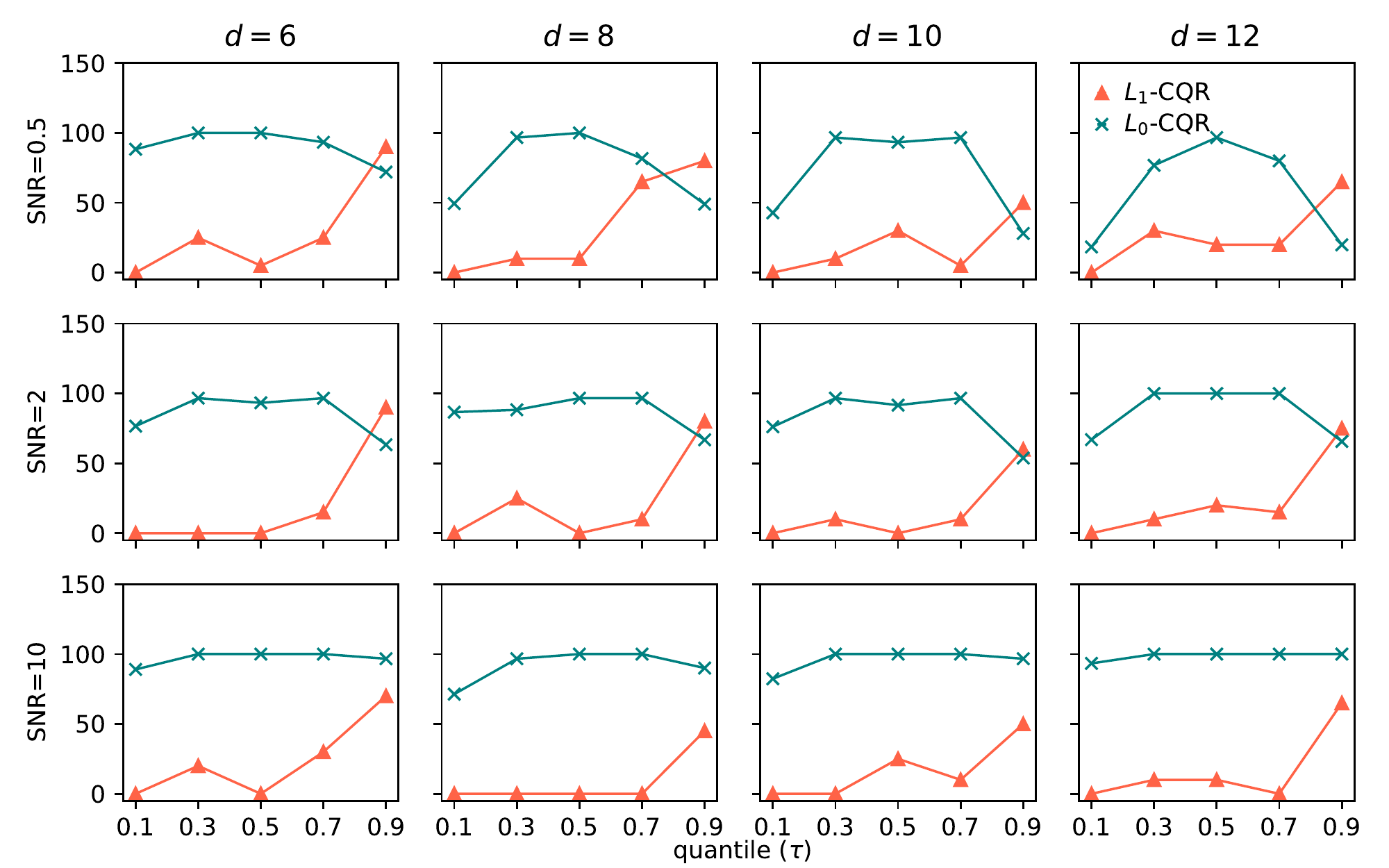}
    \caption{Accuracies of the $\mathcal{L}_1$-CQR and $\mathcal{L}_0$-CQR approaches with $n = 500$ and $k_{\text{true}} = 2$.}
    \label{fig:figb5}
    \vspace{-1em}
\end{figure}

\end{document}